\documentclass[12pt,reqno]{amsart}
\usepackage{amssymb}
\usepackage{hyperref}
\usepackage{color}
\usepackage{amsmath}
\usepackage{amsthm}
\usepackage{lineno}

\newtheorem{theorem}{Theorem}[section]
\newtheorem{definition}[theorem]{Definition}
\newtheorem{proposition}[theorem]{Proposition}
\newtheorem{lemma}[theorem]{Lemma}
\newtheorem{corollary}[theorem]{Corollary}

\newtheorem{problem}[theorem]{Problem}
\newtheorem{remark}[theorem]{Remark}

\title{Positive reduction from spectra}
\date{\today}
\author{Maria Anastasia Jivulescu}
\address{M.A.J.: Department of Mathematics,
Politehnica University of Timi\c soara,
Victoriei Square 2, 300006 Timi\c soara, Romania}
\email{maria.jivulescu@upt.ro}

\author{Nicolae Lupa}
\address{N.L.: Department of Mathematics,
Politehnica University of Timi\c soara,
Victoriei Square 2, 300006 Timi\c soara, Romania}
\email{nicolae.lupa@upt.ro}

\author{Ion Nechita}
\address{I.N.: Zentrum Mathematik, M5, Technische Universit\"at M\"unchen, Boltzmannstrasse 3, 85748 Garching, Germany
 and CNRS, Laboratoire de Physique Th\'{e}orique, IRSAMC, Universit\'{e} de Toulouse, UPS, F-31062 Toulouse, France}
\email{nechita@irsamc.ups-tlse.fr}

\author{David Reeb}
\address{D.R.: Zentrum Mathematik, M5, Technische Universit\"at M\"unchen, Boltzmannstrasse 3, 85748 Garching, Germany}
\email{david.reeb@tum.de}


\begin{document}

\begin{abstract}
We study the problem of whether all bipartite quantum states having a prescribed spectrum remain positive under the reduction map applied to one subsystem. We provide necessary and sufficient conditions, in the form of a family of linear inequalities, which the spectrum has to verify. Our conditions become explicit when one of the two subsystems is a qubit, as well as for further sets of states. Finally, we introduce a family of simple entanglement criteria for spectra, closely related to the reduction and positive partial transpose criteria, which also provide new insight into the set of spectra that guarantee separability or positivity of the partial transpose.
\end{abstract}

\maketitle

\section{Introduction}

One of the most studied problems in quantum information theory is
to find methods to decide whether a given quantum state is separable or entangled \cite{horodeckireview}.
We recall that a quantum state $\rho\in M_n(\mathbb C)\otimes M_k(\mathbb C)$  (here  $M_n(\mathbb C)$ denotes the space of all $n\times n$ complex matrices) is called \emph{separable} \cite{wernerentangled} if it can be written as
$$\rho=\sum_i p_i e_ie_i^* \otimes f_if_i^*$$
with $p_i\geq 0$, $\sum_i p_i=1$, $e_i\in\mathbb{C}^n$, $f_i\in\mathbb{C}^k$ (throughout the paper we will identify states with their density matrices). States which are not separable are called \emph{entangled}. Note that the set of separable states ($\mathrm{SEP}$) is a convex subset of the convex set of all states. The extremal points of $\mathrm{SEP}$ are the pure product states, i.e.\ tensor products of one-dimensional projectors.

The separability  problem has been proved to be $NP$-hard \cite{Gu.2003-proc}.
It can be mathematically related to positive maps on $C^*$-algebras since a quantum state $\rho\in M_n(\mathbb C)\otimes M_k(\mathbb C)$ is separable if and only if $(\mathrm{id}_n\otimes P)(\rho)$ is positive-semidefinite for all positive maps $P:M_k(\mathbb C)\to M_m(\mathbb C)$ and all positive integers $m \in \mathbb N$,  where $\mathrm{id}_n$ is the identity map on some matrix algebra with appropriate dimension (here, $n$) \cite{horodeckiPPTcriterion}. Thus, each fixed positive map gives a necessary condition for separability. For example, the \emph{positive partial transpose} ($\mathrm{PPT}$) criterion corresponds to the choice $P=\Theta$, where $\Theta$ denotes the transposition map on $M_k(\mathbb C)$. Moreover,  the $\mathrm{PPT}$ criterion is also sufficient for $nk\leq 6$  \cite{woro,horodeckiPPTcriterion}, but this equivalence is wrong in higher dimensions.

An alternative choice of the positive map $P$ is the \emph{reduction map}
$$R:M_k(\mathbb C)\to M_k(\mathbb C), \; R(X):=I_k\cdot\mathrm{Tr}[X]-X,$$
and the corresponding separability test is called  \emph{reduction $(\mathrm{RED})$ criterion}  \cite{cerf,horodeckireduction}.
The reduction criterion is weaker than the $\mathrm{PPT}$ criterion: if a state violates the reduction criterion, then it also violates the $\mathrm{PPT}$ criterion \cite{horodeckireduction}. Conversely, there exist states (some entangled Werner states \cite{wernerentangled}) which satisfy the reduction criterion but violate the $\mathrm{PPT}$ criterion. On the other hand, the two criteria are equivalent if the subsystem on which the reduction map is applied is a qubit \cite{cerf}.
The importance of the reduction criterion stems from its connection to entanglement distillation \cite{horodeckireduction}: any state which violates the reduction criterion is distillable. Recall that a bipartite entangled state is distillable if a pure maximally entangled state can be obtained arbitrarily closely, by local quantum operations and classical communication, from many copies of that state.

A possible approach  to the separability problem is to study  \emph{absolutely separable states} (ASEP), i.e.\  states  that
remain separable under any global unitary transformation  \cite{kus}. Since absolute separability is a purely spectral property, the problem is to find conditions on the spectrum that characterize absolutely separable states, i.e.\ to find constrains on the eigenvalues of a state $\rho$ guaranteeing that $\rho$ is separable with respect to any decomposition of the corresponding product tensor space \cite{knil}. This problem was first fully solved in the qubit-qubit case in \cite{verstraeteaudenaert}.
Furthermore, it is known that there is a ball of known Euclidean radius centered at the maximally-mixed state $\frac{1}{nk} (I_n\otimes I_k)$ such that every state within this ball is separable \cite{largestseparableball} (see also \cite{zycz}), meaning that any state within this ball is actually absolutely separable. However, there exist absolutely separable states outside of this ball \cite[Appendix B]{vidaltarrach} (cf.\ Remark \ref{ASEPoutsideSEPBALL}).
In analogy to absolutely separable states, states which remain $\mathrm{PPT}$ under any global unitary transformation are called  \emph{absolutely $\mathrm{PPT}$ states} ($\mathrm{APPT}$) \cite{zycz}.
Necessary and sufficient conditions on the spectrum of these states are given in \cite{hil}, in the form of a finite set (albeit exponentially large in the dimension) of linear matrix inequalities. Finally, it was shown that in the qudit-qubit case ($\mathbb{C}^n\otimes\mathbb{C}^2$ quantum systems) the set of  absolutely $\mathrm{PPT}$ states coincides with the set of absolutely separable states  \cite{joh}, meaning that one also has a finite necessary and sufficient criterion for absolute separability in the case where one of the subsystems is a qubit.

In this paper, we introduce and characterize the set of  \emph{absolutely $\mathrm{RED}$ states}, i.e.  states which remain positive under the reduction map ($\mathrm{RED}$) applied to one subsystem  after any global unitary transformation. Our main result (Theorem \ref{thm:ared}) provides  a necessary and sufficient condition on the spectrum under which a state is absolutely $\mathrm{RED}$. This condition can be stated in the form of a family of linear inequalities in terms of the spectrum of the reduction
of a pure state given by its Schmidt coefficients (Theorem \ref{thm:red-pure}). Moreover, we obtain an explicit criterion for pseudo-pure states to be absolutely $\mathrm{RED}$ (Proposition \ref{exampleprop}). We also provide
simple polyhedral approximations of the set of absolutely $\mathrm{RED}$ states by establishing upper and lower bounds on it (Theorem \ref{intermediatecriteriathm}).
Additionally,  a linear sufficient condition for a state to be absolutely $\mathrm{PPT}$  is obtained in Theorem \ref{simplesufftheoremforAPPT}, which is simpler than Hildebrand's condition \cite{hil} which consists in checking the positivity of an exponential number of Hermitian matrices. As a consequence, we deduce a lower bound for the set of  absolutely $\mathrm{PPT}$ states.

\medskip

\noindent\emph{Note added:} After completion of the present work, we became aware of the recent paper \cite{johnstonmay2014}, which investigates the relationship between the set of absolutely separable states and the set of absolutely PPT states, providing evidence for the conjecture $\mathrm{ASEP} = \mathrm{APPT}$. The content of our Proposition \ref{boundaryproposition} is implicit in the proof of \cite[Proposition 1]{johnstonmay2014}.

\medskip

\noindent \textit{Acknowledgments.} We would like to thank Marco Piani for making us aware of Ref.\ \cite{marcopiani} (see around Proposition \ref{marcoremark} below).
The authors also like to thank to the referee for his/hers useful comments and suggestions meant to improve the quality of our paper.

The work of M.A.J. and N.L. was supported by a grant of the Romanian National Authority for Scientific Research, CNCS-UEFISCDI, project number  PN-II-ID-JRP-RO-FR-2011-2-0007.
I.N.'s research has been supported by a von Humboldt fellowship and by the ANR projects {OSQPI} {2011 BS01 008 01} and {RMTQIT}  {ANR-12-IS01-0001-01}. D.R. acknowledges support from an EU Marie Curie grant, number 298742 QUINTYL.

\section{The absolute reduction criterion}\label{preliminariessection}
The set of density operators (positive-semidefinite matrices of unit trace) acting on $\mathbb C^d$ is denoted by $D_d$. In this work we will mostly be concerned with bipartite quantum systems represented on a tensor product Hilbert space $\mathbb C^n\otimes \mathbb C^k\equiv\mathbb C^{nk}$, and we denote the set of quantum states on such a bipartite system also by $D_{n,k}$ with the subscripts indicating the bipartition. Except for Sections \ref{sectiondecompofdifferentdimensions} and \ref{remarkssection}, $n$ will denote the Hilbert space dimension of the first tensor factor and $k$ that of the second one. We will always take $n,k\geq2$ as the questions become trivial otherwise.

We denote the set of separable states \cite{wernerentangled,horodeckireview} in $D_{n,k}$ by
$$\mathrm{SEP}_{n,k}:=\{\rho\in D_{n,k}\,|\,\rho~\text{separable}\}.$$
A central goal in quantum information theory is to find upper and lower approximations to $\mathrm{SEP}_{n,k}$ \cite{horodeckireview}.

On any matrix algebra $M_d(\mathbb C)$, we define the \emph{reduction map}
\begin{align*}
R:M_d(\mathbb C)\to M_d(\mathbb C), \qquad R(X):=I_d\cdot\mathrm{Tr}[X]-X,
\end{align*}
where $I_d\in M_d(\mathbb{C})$ is the identity matrix of size  $d$ and $\mathrm{Tr}$ is the usual, unnormalized, matrix trace. From the definition, it follows that the map $R$ is positive, i.e. $R(X) \geq 0$ whenever $X \geq 0$.
We write the transposition map on any matrix algebra $M_d(\mathbb C)$ as $\Theta$, and we also write $\Theta(X)\equiv X^T$. We point out that the reduction map $R$ is completely co-positive, i.e. $R\Theta : X\mapsto I_d\cdot\mathrm{Tr}[X]-\Theta(X)$ is a completely positive map \cite{tomyiama}.

For a bipartite matrix $X= X_{AB} \in M_n(\mathbb C) \otimes M_k(\mathbb C)\equiv M_{nk}(\mathbb C)$, its reduction over the second subsystem ($B$) is denoted by
$$X^{red} := ({\rm id}_n\otimes R)(X_{AB}) = X_A \otimes I_k - X_{AB},$$
where $X_A: = (\mathrm{id}_n \otimes \mathrm{Tr})(X)$ denotes the partial trace over $(B)$ of the matrix $X=X_{AB}$ (cf.\ \cite{nielsenchuang} for these general notions). We denote the reduction over the first subsystem ($A$) by
$$X^{red'}:=(R\otimes{\rm id}_k)(X_{AB})=I_{n}\otimes X_B-X_{AB}.$$
We denote the \emph{partial transposition} of a bipartite matrix $X=X_{AB}\in M_n(\mathbb C) \otimes M_k(\mathbb C)$ by
$$X^\Gamma:=({\rm id}_n\otimes\Theta)(X).$$

%The composition of $\Theta$ with the completely positive map $R\Theta:X\mapsto I\cdot\mathrm{Tr}[X]-\Theta(X)$ is the reduction map $R$ defined above; one says that the reduction map $R$ is completely co-positive.

As is well known, every positive map $P$ on $M_k(\mathbb C)$ defines an entanglement criterion \cite{horodeckiPPTcriterion,horodeckireview}: if, for $\rho\in D_{n,k}$, the matrix $({\rm id}_n\otimes P)(\rho)$ is not positive-semidefinite, then $\rho$ is entangled. Specializing to the reduction map $P=R$, this becomes the \emph{reduction criterion} \cite{horodeckireduction,cerf}, which is also related to the distillability of the state in question \cite{horodeckireview}. Every bipartite state whose entanglement is detected by the reduction criterion is also detected by the partial transposition criterion \cite{perescriterion,horodeckiPPTcriterion}, which is the above criterion for the map $P=\Theta$; this follows due to the representation of $R$ as the composition of $\Theta$ with the completely positive map $X\mapsto I_k\cdot\mathrm{Tr}[X]-\Theta(X)$.

The set of density operators $\rho\in D_{n,k}$ having positive reductions with respect to the second resp.\ first tensor factor for the fixed tensor decomposition $M_{nk}(\mathbb C) \cong M_n(\mathbb C) \otimes M_k(\mathbb C)$ is denoted by
$$\mathrm{RED}_{n,k}:= \{\rho \in D_{n,k} \,|\, \rho^{red} \geq 0\}\qquad\text{and}\qquad\mathrm{RED}'_{n,k}:=\{\rho \in D_{n,k} \,|\, \rho^{red'} \geq 0\}.$$
Moreover, we shall denote by
$$\mathrm{RED}_{n,k}'' := \mathrm{RED}_{n,k} \cap \mathrm{RED}_{n,k}' = \{\rho \in D_{n,k} \,|\, \rho^{red} \geq 0 \text{ and } \rho^{red'} \geq 0\},$$ the set of density matrices which have \emph{both} reductions positive. The above described entanglement criterion \cite{horodeckiPPTcriterion} implies the inclusion $\mathrm{SEP}_{n,k}\subseteq\mathrm{RED}''_{n,k}$ \cite{horodeckireduction,cerf}. Recall also that the set of states with positive partial transpose is
$$\mathrm{PPT}_{n,k}:=\{\rho\in D_{n,k}\,|\,\rho^\Gamma\geq0\}.$$
Note that, when $k=2$, the reduction and the PPT criterion are equivalent \cite{horodeckireduction,cerf,jln}, i.e.\ they detect entanglement for the same states, so that $\mathrm{PPT}_{n,2}=\mathrm{RED}_{n,2}$. Furthermore, it is well known that $\mathrm{SEP}_{n,k}=\mathrm{PPT}_{n,k}$ whenever $nk\leq6$ \cite{horodeckiPPTcriterion}. Occasionally we will write $\mathrm{RED}$ instead of $\mathrm{RED}_{n,k}$ etc., as the dimensions of the subsystems will be clear from the context most of the time.

We introduce the $(d-1)$-dimensional \emph{probability simplex}:
$$\Delta_d~:=~\{x = (x_1,\ldots,x_d)\in\mathbb R^d\,|\,\forall i:x_i\geq0,\,\sum_{i=1}^dx_i=1\}\,.$$
Any permutation-invariant set $A\subseteq\Delta_d$ defines the subset $\tilde A:=\{\rho\in D_d\,|\,{\rm spec}\left(\rho\right)\in A\}$ of all density matrices whose spectrum lies in $A$ (including multiplicities of eigenvalues; here we identify $\mathrm{spec}\left(\rho\right)$ as the vector of eigenvalues of $\rho$). Conversely, any set $\tilde A\subseteq D_{d}$ which is invariant under all unitary conjugations, i.e.\ $U\tilde AU^*=\tilde A$ $\forall$ unitaries $U\in M_d(\mathbb C)$, can be uniquely identified with a set of spectra $A\subseteq\Delta_d$. Throughout the paper, we freely identify $A\equiv\tilde A$ for such subsets of quantum states for which membership is decided by spectral information alone.

\medskip

Starting from the subsets of bipartite quantum states introduced above, we now define special such spectral sets:
\begin{definition}\label{def:ARED}\rm
The set of states which remain $\mathrm{RED}$ (i.e.\ positive under the reduction map applied to the second tensor factor) under any global unitary transformation $U\in\mathcal U_{nk}$ is denoted by $\mathrm{ARED}$ (``\emph{absolutely} $\mathrm{RED}$''):
$$\mathrm{ARED}_{n,k}~:=~\{ \rho \in D_{n,k} \,|\, \forall U \in \mathcal U_{nk}: \, (U\rho U^*)^{red} \geq 0 \} = \bigcap_{U \in \mathcal U_{nk}} U \mathrm{RED}_{n,k} U^*.$$
Similarly:
\begin{align*}
\mathrm{ARED}'_{n,k}~&:=~\bigcap_{U \in \mathcal U_{nk}} U \mathrm{RED}'_{n,k} U^*\,,\\
\mathrm{ARED}''_{n,k}~&:=~\bigcap_{U \in \mathcal U_{nk}} U \mathrm{RED}''_{n,k} U^*=\mathrm{ARED}_{n,k} \cap \mathrm{ARED}'_{n,k}\,,\\
\mathrm{APPT}_{n,k}~&:=~\bigcap_{U\in\mathcal U_{nk}}U\mathrm{PPT}_{n,k}U^*\,,\\
\mathrm{ASEP}_{n,k}~&:=~\bigcap_{U\in\mathcal U_{nk}}U\mathrm{SEP}_{n,k}U^*\,.
\end{align*}
\end{definition}
The fact that $\rho^\Gamma=({\rm id}_n\otimes\Theta)(\rho)$ and $(\Theta\otimes{\rm id}_k)(\rho)$ have the same spectrum implies, together with identifying $\mathrm{APPT}_{n,k}$ as a subset of $\Delta_{nk}=\Delta_{kn}$ as described above, that $\mathrm{APPT}_{n,k}=\mathrm{APPT}_{k,n}$; similarly, $\mathrm{ASEP}_{n,k}=\mathrm{ASEP}_{k,n}$. The set $\mathrm{ARED}_{n,k}$ does generally not share this invariance as the dimension of the subsystem to which the reduction map is applied does matter, see Section \ref{sectiondecompofdifferentdimensions} and also the more explicit examples in Sections \ref{qubitsection} and \ref{examplesection}; it is however true, from the definition, that $\mathrm{ARED}_{n,k}''=\mathrm{ARED}_{k,n}''$.

More generally than in Definition \ref{def:ARED}, we may define for any subset $C\subseteq D_{n,k}$ the set
$$\mathrm{A}C:=\{ \rho \in D_{n,k} \,|\, \forall U \in \mathcal U_{nk}: \, U\rho U^*\in C \} = \bigcap_{U \in \mathcal U_{nk}} U C U^*.$$
Then we get the following:
\begin{lemma}\label{lemmaAC}
Let $C\subseteq D_{n,k}$ be a convex set. If $\rho\in\mathrm{A}C$ majorizes $\sigma\in D_{n,k}$, i.e.\ if $\sigma\prec\rho$, then $\sigma\in\mathrm{A}C$.
\end{lemma}
\begin{proof}By the quantum generalization of Birkhoff's Theorem for majorization \cite{uhlmann}, there exist unitary matrices $U_j\in \mathcal U_{nk}$ and a probability distribution $\{p_j\}$ such that $$\sigma=\sum_jp_jU_j\rho U_j^*.$$ Now, $\rho\in\mathrm{A}C=\bigcap_{U \in \mathcal U_{nk}} UCU^*$ implies $U_j\rho U_j^*\in\mathrm{A}C$ for all $j$. Since $\mathrm{A}C$ is convex as an intersection of convex sets $UCU^*$, we have $\sigma=\sum_jp_jU_j\rho U_j^*\in\mathrm{A}C$.
\end{proof}

Thus, the set $\mathrm{ASEP}_{n,k}$  is ``majorization-invariant'', since $\mathrm{SEP}_{n,k}$ is convex by definition; the same reasoning holds for the sets $\mathrm{ARED}_{n,k}$ and $\mathrm{APPT}_{n,k}$. See also Lemma \ref{lemmaAZcirc} for a proof using another characterization.

Finally, we introduce some general notation. We denote $[n]:=\{1,2,\ldots,n\}$. For any vector $\lambda\in\mathbb R^d$, we denote by $\lambda^\uparrow\in\mathbb R^d$ the vector having the same entries ordered increasingly, i.e.\ $\lambda^\uparrow_1\leq\lambda^\uparrow_2\leq\ldots\leq\lambda^\uparrow_d$; similarly, we define the decreasingly-ordered vector $\lambda^\downarrow$.

\section{Reductions of pure states}

The main ingredient in the proof of our main contribution, Theorem \ref{thm:ared}, is the following result, giving the spectrum of the reduction of a pure state in terms of its Schmidt coefficients \cite{nielsenchuang}.

\begin{theorem}\label{thm:red-pure}
Let $\psi \in \mathbb C^n \otimes \mathbb C^k$ be a vector having Schmidt decomposition
$$\psi = \sum_{i=1}^r \sqrt{x_i} e_i \otimes f_i,$$
where $r \leq \min(n,k)$ is the Schmidt rank of $\psi$, $x_i >0$, and $(e_i)_{i=1}^{n}$, $(f_i)_{i=1}^{k}$ are orthonormal families in $\mathbb C^n$ and $\mathbb C^k$, respectively. If the set of Schmidt coefficients $\{x_i\}_{i=1}^r$ is equal to $\{x_1 > x_2 > \cdots > x_q\}$ and the $x_i$ have multiplicities $m_i$ $(i=1,\ldots, q)$, the eigenvalues of the reduced projection on $\psi$ are
$$\mathrm{spec}\left((\psi\psi^*)^{red}\right)= \{x_1>\eta_1> \cdots >\eta_{q-1}> x_q > 0 \geq \eta_q\},$$
where the eigenvalues $x_i$ have multiplicities $m_ik-1$, the eigenvalues $\eta_i$ are simple and the null eigenvalue has multiplicity $(n-r)k$. The eigenvalues $\eta_i$ are the $q$ real solutions of the equation $F_x(\lambda)=0$, where
$$F_x(\lambda) := \sum_{i=1}^q \frac{m_i x_i}{x_i-\lambda} - 1.$$
Finally, if $r>1$, then $\eta_q < 0$.
\end{theorem}
\begin{proof}
First, compute
$$\tau:=(\psi\psi^*)^{red} = \sum_{i=1}^rx_ie_ie_i^* \otimes I_k - \psi\psi^*$$
and observe immediately that $\tau$ has support included in $\mathrm{span}(e_i)_{i=1}^r \otimes \mathbb C^k$, so the null space of $\tau$ has dimension at least $(n-r)k$. Moreover, notice that for all $(i,j) \in [r]\times [k]$, $i\neq j$ we have that $\tau(e_i \otimes f_j) = x_i e_i \otimes f_j$, so each of the eigenvalues $x_i$ has multiplicity $k-1$ (here, we consider the Schmidt coefficients $\{x_i\}_{i=1}^r$ with multiplicities). The above discussion completely describes the action of $\tau$ on the space $\left(\mathrm{span}(e_i \otimes f_i)_{i=1}^r\right)^\perp$.

The action of $\tau$ on $\mathrm{span}(e_i \otimes f_i)_{i=1}^r$ has the following matrix in the ``canonical'' basis $(e_i \otimes f_i)_{i=1}^r$:
$$\forall i,j \in [r], \quad M_\tau(i,j) = x_i \delta_{ij} - \sqrt{x_ix_j}.$$
The claim follows now from Lemma \ref{MXlemma}.
\end{proof}

\begin{lemma}\label{MXlemma}For $i\in[r]$, let $x_i>0$, ordered in such a way that the sets $\{x_j|j\in[r]\}$ and $\{x_1 > x_2 > \cdots > x_q\}$ equal each other, and $x_i$ has multiplicity $m_i~(i\in[q])$ . Define the matrix $M\in M_r(\mathbb R)$ with entries $M_{ij}:=x_i\delta_{ij}-\sqrt{x_ix_j}~(i,j\in[r])$. Then:
\begin{enumerate}
\item The eigenvalues of $M$ are
\begin{align*}
\forall i \in [q], \quad \lambda=x_i, \quad &\text{with multiplicity } m_i -1,\\
\forall i \in [q], \quad \lambda=\eta_i, \quad &\text{with multiplicity } 1,
\end{align*}
where $\eta_1> \cdots > \eta_q$ are the $q$ real solutions of the equation $\sum_{i=1}^q \frac{m_i x_i}{x_i-\eta}=1.$
\item\label{MXlemmatrace0}It is $x_1>\eta_1>x_2>\eta_2>\cdots>\eta_{q-1}>x_q>0\geq\eta_q=-\sum_{i=1}^{q-1}\eta_i-\sum_{i=1}^q(m_i-1)x_i$.
\item It is $\|M\|=-\eta_q\leq\frac{r-1}{r}\sum_{i=1}^rx_i$, with equality if and only if $x_i=x_j~\forall i,j\in[r]$.
\end{enumerate}
\end{lemma}
\begin{proof}(1) Let $D\in M_r(\mathbb R)$ be the positive definite diagonal matrix with entries $D_{ij}=x_i\delta_{ij}$, and $v\in\mathbb R^{r}$ be the vector with entries $v_i=\sqrt{x_i}$. The characteristic polynomial of $M$ is then $P(X)=\det[D-XI_r-vv^*]$. The matrix $(D-xI_r)$ is invertible for $x\in\mathbb R\setminus\{x_1,\ldots,x_r\}$ with Hermitian inverse, and the characteristic polynomial evaluated at $x$ is therefore:
\begin{align*}
P(x)&=\det\left[D-xI_r-vv^*\right]\\
&=\det\left[(D-xI_r)^{1/2}\left(I_r-(D-xI_r)^{-1/2}vv^*(D-xI_r)^{-1/2}\right)(D-xI_r)^{1/2}\right]\\
&=\det\left[(D-xI_r)^{1/2}\right]\det\left[I_r-(D-xI_r)^{-1/2}vv^*(D-xI_r)^{-1/2}\right]\det\left[(D-xI_r)^{1/2}\right]\\
&=\left(1-v^*(D-xI_r)^{-1}v\right)\det\left[D-xI_r\right]\\
&=\left(1-\sum_{i=1}^r\frac{x_i}{x_i-x}\right)\prod_{j=1}^r(x_j-x)~=~\prod_{i=1}^r (x_i - x) - \sum_{i=1}^r x_i \prod_{j \neq i} (x_j - x).
\end{align*}
Here, we used that $\det[I_r-ww^*]=1-w^*w$ for $w\in\mathbb R^r$. Due to continuity, this last line gives $P(x)$ actually for all $x\in\mathbb R$. The claim about the eigenvalues and their multiplicities follows now immediately. The above method for computing the eigenvalues of a rank-one perturbation to a diagonal matrix is well-known \cite{golubreview,and}, but has been repeated here for convenience.

(2) The interlacing of the eigenvalues $x_i$ and $\eta_i$ follows from the fact that the function $f(\eta):=\sum_{i=1}^q \frac{m_i x_i}{x_i-\eta}$ is strictly increasing on each of the intervals of its domain and from the following relations: $\lim_{\eta \to x_i^{-}} f(\eta) = +\infty$, $\lim_{\eta \to x_i^{+}} f(\eta) = -\infty$, $\lim_{\eta \to \pm\infty} f(\eta) = 0$, and $f(0)=r\geq1$. The expression for $\eta_q$ follows from the fact that ${\rm Tr}[M]=0$ equals the sum of all eigenvalues of $M$ (including multiplicities).

(3) For the Hermitian matrix $M$, it follows from the previous items that
$$\|M\|=
\begin{cases}
\max(\eta_1,-\eta_q)& \text{ if } m_1=1,\\
\max(x_1,-\eta_q)&  \text{ if } m_1\geq2.
\end{cases}
$$
In either case, the expression for $\eta_q$ from item (2) shows then $\|M\|=-\eta_q$.

To prove the inequality, let $w\in\mathbb R^r$ be any vector having components $w_i~(i\in[r])$. Then two applications of the Cauchy-Schwarz inequality give:
\begin{align*}
w^*Mw&=\sum_{i=1}^rx_iw_i^2-\sum_{i,j=1}^rw_i\sqrt{x_i}\sqrt{x_j}w_j=\frac{1}{r}\sum_{i=1}^r(\sqrt{x_i}w_i)^2\,\sum_{j=1}^r1-\left(\sum_{i=1}^r\sqrt{x_i}w_i\right)^2\\
&\geq\frac{1}{r}\left(\sum_{i=1}^r\sqrt{x_i}w_i\cdot1\right)^2-\left(\sum_{i=1}^r\sqrt{x_i}w_i\right)^2=-\left(1-\frac{1}{r}\right)\left(\sum_{i=1}^r\sqrt{x_i}w_i\right)^2\\
&\geq-\frac{r-1}{r}\left(\sum_{i=1}^rx_i\right)\sum_{j=1}^rw_j^2=-\frac{r-1}{r}\left(\sum_{i=1}^rx_i\right)w^*w.
\end{align*}
This shows $\eta_q\geq-\frac{r-1}{r}\sum_{i=1}^rx_i$. The equality statement follows from the equality cases in the Cauchy-Schwarz inequality.
\end{proof}

We now state the above Theorem \ref{thm:red-pure} in a form that allows for a uniform treatment of degenerate and possibly non-positive Schmidt coefficients. It is a simple restatement of results shown in the proofs of Theorem \ref{thm:red-pure} and Lemma \ref{MXlemma}.

\begin{corollary}\label{simplestatementofpurereduction}
For a vector $\psi \in \mathbb C^n \otimes \mathbb C^k$ with non-negative (but possibly non-positive) Schmidt coefficients $\{x_i\}_{i=1}^r$ (w.l.o.g.\ ordered non-increasingly), the eigenvalues of the reduced matrix $(\psi\psi^*)^{red}$ are
\begin{equation*}\label{eq:spect}
\mathrm{spec}\left((\psi\psi^*)^{red}\right) = (\underbrace{x_1, \ldots, x_1}_{k-1 \text{ times}}, \eta_1 ,  \underbrace{x_2, \ldots, x_2}_{k-1 \text{ times}}, \ldots, \eta_{r-1}, \underbrace{x_r, \ldots, x_r}_{k-1 \text{ times}}, \underbrace{0, \ldots, 0}_{(n-r)k \text{ times}},\eta_r) \in \mathbb R^{nk},
\end{equation*}
where $x_i\geq\eta_i\geq x_{i+1}$ for $i\in[r-1]$ and $\eta_r=-\sum_{i=1}^{r-1}\eta_i\leq0$. The set $\{\eta_i\}_{i=1}^r\setminus\{x_i\}_{i=1}^r$ equals the set of solutions $\eta\in\mathbb R\setminus\{x_i\}_{i=1}^r$ to the equation $\sum_{i=1}^r \frac{x_i}{x_i-\eta}=1$.
\end{corollary}

We record the following important definition and notation for later use:
\begin{definition}[``\emph{hat operation}'' $x\mapsto\hat x$]\label{def:hat-xi}\rm
Given  $n,k\geq 2$ and a vector $x \in \mathbb R^r_+$ with $r \leq \min(n,k)$, we associate to $x$ the vector $\psi \in \mathbb C^n \otimes \mathbb C^k$ given by
$$\psi = \sum_{i=1}^r \sqrt{x_i} e_i \otimes f_i,$$
where $(e_i)_{i=1}^n$, $(f_i)_{i=1}^k$ are fixed orthonormal families in $\mathbb C^n$ resp.\ $\mathbb C^k$.
We then define $\hat x$ to be the vector of eigenvalues of the reduction $(\psi\psi^*)^{red}$ of the quantum state $\psi\psi^*\in D_{n,k}$, taken with multiplicities as in Corollary \ref{simplestatementofpurereduction}:
$$\hat x ~:=~ (\underbrace{x_1, \ldots, x_1}_{k-1 \text{ times}}, \eta_1 ,  \underbrace{x_2, \ldots, x_2}_{k-1 \text{ times}}, \ldots, \eta_{r-1}, \underbrace{x_r, \ldots, x_r}_{k-1 \text{ times}}, \underbrace{0, \ldots, 0}_{(n-r)k \text{ times}},\eta_r)\,\in\,\mathbb R^{nk}.$$
We point out that vectors $x$ with repeating or null coordinates are allowed in the above construction.
%Moreover, when $x=x^\downarrow$ is decreasingly ordered, we assume an ordering such that $\hat x = \hat x ^\downarrow$ (see Corollary \ref{simplestatementofpurereduction}).
Note that the ``hat operation'' $x\mapsto\hat x$ does depend on the dimensions $n$ and $k$ and also on the convention that the reduction map is applied to the second tensor factor (corresponding to $\mathbb C^k$ here), but we will leave this dependence implicit most of the time when there is no room for confusion.
\end{definition}

\begin{remark}
Note that, if $\psi$ is entangled (i.e.\ $r>1$), then the matrix $(\psi\psi^*)^{red}$ is not positive since $\eta_r<0$. Hence, the reduction criterion detects pure entanglement.
\end{remark}

\begin{remark}\label{remarktraceremark}
From the definition of the reduction criterion, it follows that
$$\sum_{i=1}^{nk} \hat x_i = \mathrm{Tr} \left[ (\psi\psi^*)^{red} \right]= (k-1)\|\psi\|^2\,,$$
which equals $(k-1)$ if $\psi$ was a normalized vector. More generally, the reduction map $R$ applied to a $k$-dimensional (sub)-system (see Section \ref{preliminariessection}) scales the trace of any matrix by a factor of $(k-1)$.
\end{remark}

We now relate the spectrum of reduced pure states $(\psi\psi^*)^{red}$, as found in Theorem \ref{thm:red-pure} above, to the \emph{entanglement of disturbance} $Q_{D_1,\{\Pi_B\}}(\psi\psi^*)$, which was recently introduced by Piani {\it{et al.}}\ in \cite{marcopiani} for any pure state $\psi\in\mathbb C^n\otimes\mathbb C^k$ as follows:
\begin{align*}
Q_{D_1,\{\Pi_B\}}(\psi\psi^*)~:=&~\min_{(g_j)}\frac{1}{2}\Big\|\psi\psi^*-\sum_j(I_n\otimes g_jg^*_j)\psi\psi^*(I_n\otimes g_jg^*_j)\Big\|_1\,,
\end{align*}
where the minimum is taken over all orthonormal bases $(g_j)$ of $\mathbb C^k$. The entanglement of disturbance was shown to be a bona fide entanglement measure for bipartite pure states \cite{marcopiani}. Here, we relate it to the reduction map:
\begin{proposition}\label{marcoremark}For any normalized pure state $\psi\in\mathbb C^n\otimes\mathbb C^k$ we have:
\begin{align*}
Q_{D_1,\{\Pi_B\}}(\psi\psi^*)~=~\frac{\|(\psi\psi^*)^{red}\|_1\,-\,(k-1)}{2}~=~-\lambda_{min}\big((\psi\psi^*)^{red}\big)\,.
\end{align*}
\end{proposition}
\begin{proof}From the definition of the reduction map it is easy to see that, for pure states, $\lambda_{min}\big((\psi\psi^*)^{red})\leq0$ and that $(\psi\psi^*)^{red}$ has at most one negative eigenvalue (both facts are also apparent from Theorem \ref{thm:red-pure}). Thus, we have $$\|(\psi\psi^*)^{red}\|_1-{\rm Tr}\left[(\psi\psi^*)^{red}\right]=-2\lambda_{min}\big((\psi\psi^*)^{red}),$$ which together with Remark \ref{remarktraceremark} implies the second equality. Furthermore, Theorem \ref{thm:red-pure} shows that, for $x_i$ the Schmidt coefficients of $\psi$, $c=-\lambda_{min}\big((\psi\psi^*)^{red})$ is the unique non-negative root of $\sum_ix_i/(x_i+c)=1$. This agrees with the formula for $Q_{D_1,\{\Pi_B\}}(\psi\psi^*)$ derived in \cite[Theorem\ III.3]{marcopiani} and shows the remaining equality.

We now offer a more direct proof of the non-trivial fact $Q_{D_1,\{\Pi_B\}}(\psi\psi^*)=-\lambda_{min}\big((\psi\psi^*)^{red}\big),$ not using the implicit formula for either quantity. For this, note first that
\begin{align*}
Q_{D_1,\{\Pi_B\}}(\psi\psi^*)~=~\min_{(g_j)}\lambda_{max}\Big(\psi\psi^*-\sum_j(I_n\otimes g_jg^*_j)\psi\psi^*(I_n\otimes g_jg^*_j)\Big)\,,
\end{align*}
which follows from the fact the expression under the norm sign in the defining equation of $Q_{D_1,\{\Pi_B\}}(\psi\psi^*)$ is traceless with at most one positive eigenvalue \cite{marcopiani}. Furthermore, for any fixed orthonormal basis $(g_j)$ of $\mathbb C^k$, we can write
\begin{align*}
-(\psi\psi^*)^{red}~&=~\psi\psi^*-({\rm id}_n\otimes{\rm Tr})(\psi\psi^*)\,\otimes\,I_k\\
&=~\psi\psi^*-\sum_{i=1}^k(I_n\otimes g_i^*)\psi\psi^*(I_n\otimes g_i)\cdot\left(I_n\,\otimes\,\sum_{j=1}^kg_jg_j^*\right)\\
&=~\psi\psi^*-\sum_{j=1}^k(I_n\otimes g_jg_j^*)\psi\psi^*(I_n\otimes g_jg_j^*)\,-\,\sum_{i\neq j}^k(I_n\otimes g_jg_i^*)\psi\psi^*(I_n\otimes g_ig_j^*)\,.
\end{align*}
Since the last term $(-\sum_{i\neq j})$ is negative-semidefinite, we have $$-(\psi\psi^*)^{red}\leq\psi\psi^*-\sum_j(I_n\otimes g_jg^*_j)\psi\psi^*(I_n\otimes g_jg^*_j),$$ for any orthonormal basis $(g_j)$ of $\mathbb C^k$, showing that $\lambda_{max}\big(-(\psi\psi^*)^{red}\big)\leq Q_{D_1,\{\Pi_B\}}(\psi\psi^*)$.

On the other hand, when choosing $(g_j)\equiv(f_j)$ to be the orthonormal basis occurring in the Schmidt decomposition of $\psi$ (see Theorem \ref{thm:red-pure}), one easily sees the support of the term $(-\sum_{i\neq j})$ to be orthogonal to the support of $\psi\psi^*-\sum_j(I_n\otimes g_jg^*_j)\psi\psi^*(I_n\otimes g_jg^*_j)$, which is basically the observation from the first part of the proof of Theorem \ref{thm:red-pure}. This choice for $(g_j)$ thus shows $Q_{D_1,\{\Pi_B\}}(\psi\psi^*)\leq\lambda_{max}\big(-(\psi\psi^*)^{red}\big)$, and we finally get $$Q_{D_1,\{\Pi_B\}}(\psi\psi^*)=\lambda_{max}\big(-(\psi\psi^*)^{red}\big)=-\lambda_{min}\big((\psi\psi^*)^{red}\big).$$
\end{proof}

In \cite{marcopiani} also other properties of $\lambda_{min}\big((\psi\psi^*)^{red}\big)$ are derived, such as its Schur convexity as a function of the Schmidt coefficients $\{x_i\}$ of $\psi$ and upper and lower bounds depending on the largest Schmidt coefficient(s).

\section{Spectral criterion}\label{Sec:spectral}

In this section, we give a description of the set $\mathrm{ARED}_{n,k}$. We start with a technical, but easy, lemma:

\begin{lemma}\label{lem:duality-Gamma-and-red}
The partial transpose and the reduction maps are selfadjoint, which means that for both $\varphi=\Gamma$ and $\varphi=\text{red}$ we have:
$$\forall X, Y \in M_n(\mathbb C) \otimes M_k(\mathbb C), \quad \mathrm{Tr}\left[(X^\varphi)^*\,Y\right] = \mathrm{Tr}\left[X^*\,Y^\varphi\right].$$
\end{lemma}
\begin{proof}
Since both expressions are antilinear in $X$ and linear in $Y$, one can consider the case of simple tensors, $X=X_1 \otimes X_2$ and $Y=Y_1 \otimes Y_2$. With this notation, the conclusion follows by direct computation: in the case $\varphi = \Gamma$, both traces are equal to $\operatorname{Tr}\left[X_1^*Y_1\right]\operatorname{Tr}\left[X_2^*Y_2^T\right]$, while in the case $\varphi=red$, both traces are equal to $\operatorname{Tr}\left[X_1^*Y_1\right]\left(\operatorname{Tr}\left[X_2^*\right]\operatorname{Tr}\left[Y_2\right]-\operatorname{Tr}\left[X_2^*\,Y_2\right]\right)$.
\end{proof}

We now state the main result of this paper, the characterization of the set $\mathrm{ARED}_{n,k}$. The theorem follows from the rank-one case discussed in the previous section (Corollary \ref{simplestatementofpurereduction}) in a similar way as in the characterization of the APPT states \cite{hil}.

\begin{theorem}\label{thm:ared}
We have
\begin{align}\label{eq:ared}
\mathrm{ARED}_{n,k}\,=\, \{\rho \in D_{n,k} \,:\, \forall x\in \Delta_{\min(n,k)}, \, \langle \lambda_\rho^\downarrow,\hat x^\uparrow \rangle \geq 0 \},
\end{align}
where $\lambda_\rho^\downarrow$ is the vector of eigenvalues of $\rho$ ordered decreasingly and $\hat x^\uparrow$ is the increasingly ordered version of $\hat x$ that has been introduced in Definition \ref{def:hat-xi}.
\end{theorem}
\begin{proof}
A quantum state $\rho \in D_{nk}$ having ordered eigenvalues $\lambda_\rho^\downarrow$ is an element of $\mathrm{ARED}$ if and only if the following chain of equivalent statements is true:
\begin{align*}
\forall U \in \mathcal U_{nk}: \qquad &(U\rho U^*)^{red} \geq 0\\
\forall U \in \mathcal U_{nk},\, \forall \psi \in \mathbb C^{nk}, \|\psi\|=1: \qquad &\operatorname{Tr}\left[(U\rho U^*)^{red} \, \psi\psi^* \right]\geq 0\\
\forall U \in \mathcal U_{nk},\, \forall \psi \in \mathbb C^{nk}, \|\psi\|=1: \qquad &\operatorname{Tr}\left[U\rho U^* \, (\psi\psi^*)^{red} \right]\geq 0\\
\forall \psi \in \mathbb C^{nk}, \|\psi\|=1: \qquad &\min_{U \in \mathcal U_{nk}}\operatorname{Tr}\left[U\rho U^* \, (\psi\psi^*)^{red} \right]\geq 0\\
\forall x \in \Delta_{\min(n,k)}: \qquad &\langle \lambda_\rho^\downarrow,\hat x^\uparrow \rangle\geq 0,
\end{align*}
and the proof is complete. We have used Lemma \ref{lem:duality-Gamma-and-red} with $\varphi=red$, and for the last equivalence we have employed the fact (see \cite{bhatia}) that, for any selfadjoint matrices $A,B$,
$$\min_{U \in \mathcal U_{nk}}\operatorname{Tr}[AUBU^*] = \langle \lambda_A^\downarrow, \lambda_B^\uparrow \rangle,$$
where $\lambda_{A}$, $\lambda_{B}$ denote the spectra of $A$ and $B$, respectively, together with the fact that the eigenvalues of $(\psi\psi^*)^{red}$ only depend on the Schmidt coefficient vector $x\in\Delta_{\min(n,k)}$ of $\psi$ (see Corollary \ref{simplestatementofpurereduction} and Definition \ref{def:hat-xi}).
\end{proof}

\begin{remark}\label{exampleconstraintremark}
For any vector $x$ of unit rank, the condition $\langle \lambda_\rho^\downarrow,\hat x^\uparrow \rangle \geq 0$ is automatically satisfied, since, in that case, the vector $\hat x$ is positive. Hence, in the characterization Eq.\ \eqref{eq:ared} above, one can assume the vectors $x$ to have at least two non-zero components.
\end{remark}

\begin{corollary}\label{cor.ared}
A necessary condition for $\rho\in D_{n,k}$ with eigenvalue vector $\lambda\equiv\lambda_\rho$ to be an element of $\mathrm{ARED}_{n,k}$ is:
\begin{align*}
(r-1)\lambda^\downarrow_1~\leq~\sum_{i=(n-r)k+2}^{nk}\lambda^\downarrow_i\qquad\forall r~\text{with}~1\leq r\leq\min(n,k)~.
\end{align*}
\end{corollary}
\begin{proof}
For a vector $x$ of rank $r$ $(1\leq r\leq\min(n,k))$ with degenerate non-zero entries $x_1=\cdots=x_r=1/r$ it is {\rm(}see in particular Lemma \ref{MXlemma}{\rm)}: $\hat x^\uparrow_1=-(r-1)/r$, $\hat x^\uparrow_i=1/r$ for $(n-r)k+2\leq i\leq nk$, and $\hat x^\uparrow_i=0$  for the other values of $i$. The conclusion follows then from Theorem \ref{thm:ared}.
\end{proof}

We end this section with a general remark about the description of the set $\mathrm{ARED}$ obtained in Theorem \ref{thm:ared}. Consider any subset $Z \subseteq \mathbb R^{nk}$ and define the following set of quantum states (cf.\ Definition \ref{def:ARED} together with Theorem \ref{thm:ared}, and recall the notation $Z^\circ\subseteq\mathbb R^{nk}$ for the polar of the set $Z$):
$$\mathrm{A}Z^\circ := \{\rho \in D_{n,k} \,:\, \forall z \in Z, \, \langle \lambda_\rho^\downarrow,z^\uparrow \rangle \geq 0 \}.$$
Obviously, the characterization \eqref{eq:ared} is of this form, with $Z = \{\hat x \, : \, x \in \Delta_{\min(n,k)}\}$. This is also the case for the set $\mathrm{APPT}$, with $Z=\{E(x) \, : \, x \in \Delta_{\min(n,k)}\}$, see \cite[Section II]{hil}. Any such set $\mathrm{A}Z^\circ$ satisfies the following:
\begin{lemma}\label{lemmaAZcirc}
Let $Z \subseteq \mathbb R^{nk}$. If $\rho\in\mathrm{A}Z^\circ$ majorizes $\sigma\in D_{n,k}$, i.e.\ if $\sigma\prec\rho$, then $\sigma\in\mathrm{A}Z^\circ$.
\end{lemma}
\begin{proof}
By the definition of matrix majorization and by Birkhoff's Theorem \cite{bhatia}, there exist permutation matrices $P_j\in M_{nk}(\mathbb R)$ and a probability distribution $\{p_j\}$ such that $\lambda_\sigma^\downarrow=\sum_j p_jP_j\lambda_\rho^\downarrow$. The claim follows with the commutative version of a fact used already in the proof of Theorem \ref{thm:ared}, namely $\langle P_j\lambda_\rho^\downarrow,z^\uparrow\rangle\geq\langle\lambda_\rho^\downarrow,z^\uparrow\rangle$ for all $z\in\mathbb R^{nk}$ \cite{bhatia}.
\end{proof}

This gives a different proof (using the dual picture) than via Lemma \ref{lemmaAC} that the sets $\mathrm{ARED}_{n,k}$ and $\mathrm{APPT}_{n,k}$ are ``majorization-invariant''.

\section{Qubit-qudit systems}\label{qubitsection}

In this section, we consider the simplest non-trivial systems, where one of the subsystems is a qubit.

Let us start with the case when the second subsystem, i.e.\ the one on which the reduction map acts, is a qubit ($k=2$). Although the following explicit characterization of $\mathrm{ARED}_{n,2}$ is a consequence of two known results about qudit-qubit systems, i.e. $\mathbb{C}^n\otimes\mathbb{C}^2$ quantum systems (see below the proof), we derive it directly using the results in this paper.

\begin{proposition}\label{firstqubitproposition}
Let $\rho \in M_n(\mathbb C) \otimes M_2(\mathbb C)$ be a quantum state having eigenvalues $\lambda_1 \geq \cdots \geq \lambda_{2n} \geq 0$. The following are equivalent:
\begin{enumerate}
\item $\rho \in \mathrm{ASEP}_{n,2}$;
\item $\rho \in \mathrm{APPT}_{n,2}$;
\item $\rho \in \mathrm{ARED}_{n,2}$;
\item $\lambda_1 \leq \lambda_{2n-1} + 2 \sqrt{\lambda_{2n-2}\lambda_{2n}}$.
\end{enumerate}
\end{proposition}
\begin{proof}
We shall establish only the equivalence of the last two statements, since the equivalence of (1) and (2) has been shown recently in \cite{joh} and the equivalence between (2) and (3) follows from the unitary equivalence of the reduction and transposition maps on qubit systems \cite{horodeckireduction,cerf} (see also Section \ref{preliminariessection}).

Consider an arbitrary unit vector $x \in \mathbb C^n \otimes \mathbb C^2$. Its Schmidt coefficient vector is then $x = (a, 1-a)$, with $a \in [1/2,1]$. A direct computation using the formulas in Theorem \ref{thm:red-pure} gives
$$\hat x^\uparrow = (-\sqrt{a(1-a)}, \underbrace{0, \ldots, 0}_{2(n-2) \text{ times}}, 1-a, \sqrt{a(1-a)}, a).$$
From Theorem \ref{thm:ared}, it follows that $\rho \in \mathrm{ARED}_{n,2}$ if and only if
$$\forall a \in [1/2,1], \quad -\sqrt{a(1-a)}\lambda_1 + (1-a)\lambda_{2n-2} + \sqrt{a(1-a)}\lambda_{2n-1} + a \lambda_{2n} \geq 0.$$
This, in turn, is equivalent to
$$\forall a \in [1/2,1], \quad \lambda_1 - \lambda_{2n-1} \leq \sqrt{\frac{1-a}{a}}\lambda_{2n-2} + \sqrt{\frac{a}{1-a}}\lambda_{2n}.$$
Classical analysis shows that the right-hand side above achieves a minimum value of $2\sqrt{\lambda_{2n-2}\lambda_{2n}}$, finishing the proof.
\end{proof}
The explicit characterization (4) of $\mathrm{ARED}_{n,2}$ follows also from the equality $\mathrm{ARED}_{n,2}=\mathrm{APPT}_{n,2}$ \cite{horodeckireduction,cerf} together with the explicit characterization of $\mathrm{APPT}_{n,2}$ \cite{hil} (see also \cite{verstraeteaudenaert} for the cases $n=2,3$).

The qubit-qudit case ($\mathbb{C}^2\otimes\mathbb{C}^k$ quantum systems) does not follow from previous results:
\begin{proposition}
Let $\rho \in M_2(\mathbb C) \otimes M_k(\mathbb C)$ be a quantum state having eigenvalues $\lambda_1 \geq \cdots \geq \lambda_{2k} \geq 0$. Then, $\rho \in \mathrm{ARED}_{2,k}$ if and only if
$$\lambda_1 \leq \lambda_{k+1} + 2 \sqrt{(\lambda_2 + \cdots + \lambda_k)(\lambda_{k+2} + \cdots + \lambda_{2k})}.$$
\end{proposition}
\begin{proof}
The proof is almost identical to the one of the $ n \times 2$ case. For an arbitrary unit vector $x \in \mathbb C^2 \otimes \mathbb C^k$ having Schmidt coefficients $(a, 1-a)$, with $a \in [1/2,1]$, we have
$$\hat x^\uparrow = (-\sqrt{a(1-a)}, \underbrace{1-a, \ldots, 1-a}_{k-1 \text{ times}}, \sqrt{a(1-a)}, \underbrace{a, \ldots, a}_{k-1 \text{ times}}).$$
The conclusion follows from the same analysis as before.
\end{proof}

\begin{remark}
Let $n \geq 2$. Since
$$\lambda_{2n-1} + 2 \sqrt{\lambda_{2n-2}\lambda_{2n}} \leq \lambda_{n+1} + 2 \sqrt{(\lambda_2 + \cdots + \lambda_n)(\lambda_{n+2} + \cdots + \lambda_{2n})},$$
we have
$\mathrm{ARED}_{n,2} \subseteq \mathrm{ARED}_{2,n}$. See Corollary \ref{cor:nk-kn} for a more general statement.
\end{remark}

\section{Pseudo-pure states}\label{examplesection}

In this section, we study a special class of quantum states on $\mathbb C^n\otimes\mathbb C^k$, namely those lying on the segment between the ``central point'' of the set of states, $I_{nk}/(nk)$ and an extremal point of $D_{n,k}$, a pure state $v$. This family of states has been termed the ``pseudo-pure states'' following their introduction in NMR quantum computing in \cite{coryetal}. For $v\in\mathbb C^n\otimes \mathbb C^k$ with $\|v\|=1$ and $\mu\in[0,1]$, we evaluate here whether or not the state $$\rho_{v,\mu}:=\mu I_{nk}/(nk)+(1-\mu)vv^*$$ is in $\mathrm{(A)RED}$ or $\mathrm{(A)PPT}$, obtaining explicit criteria in all cases.

\begin{proposition}\label{exampleprop}
Let $n,k\geq2$, $v\in\mathbb C^n\otimes \mathbb C^k$ with $\|v\|=1$, and $\mu\in[0,1]$. Define $$\rho_{v,\mu}:=\mu I_{nk}/(nk)+(1-\mu)vv^*, \; r:=\min(n,k),$$ and denote by $\nu_1\geq\nu_2\geq\cdots\geq\nu_r\geq0$ the Schmidt coefficients of $v$ {\rm(}with the convention $\sum_{i=1}^r\nu_i=1${\rm)}. Then:
\begin{enumerate}
\item $\rho_{v,\mu}\in\mathrm{ARED}_{n,k} \quad \Leftrightarrow \quad \mu\geq\left(\frac{k-1}{nk}\frac{r}{r-1}+1\right)^{-1}$.
\item $\rho_{v,\mu}\in\mathrm{PPT}_{n,k} \quad \Leftrightarrow \quad \rho_{v,\mu}\in\mathrm{SEP}_{n,k} \quad \Leftrightarrow\quad \mu\geq\left(\frac{1}{nk\sqrt{\nu_1\nu_2}}+1\right)^{-1}$.
\item $\rho_{v,\mu}\in\mathrm{APPT}_{n,k} \quad \Leftrightarrow \quad  \rho_{v,\mu}\in\mathrm{ASEP}_{n,k} \quad \Leftrightarrow\quad \mu\geq\left(\frac{2}{nk}+1\right)^{-1}$.
\item $\rho_{v,\mu}\in\mathrm{RED}_{n,k} \quad \Leftrightarrow \quad \sum_{i=1}^r\left(\frac{1}{\nu_i}\frac{\mu}{1-\mu}\frac{k-1}{nk}+1\right)^{-1}\leq1$.
\end{enumerate}
\end{proposition}
\begin{proof}(1)  For $x\in\Delta_{\min(n,k)}$, we use the notation from Definition \ref{def:hat-xi} with  $r=\min(n,k)$, i.e.
$$\hat x ~:=~ (\underbrace{x_1, \ldots, x_1}_{k-1 \text{ times}}, \eta_1 ,  \underbrace{x_2, \ldots, x_2}_{k-1 \text{ times}}, \ldots, \eta_{r-1}, \underbrace{x_r, \ldots, x_r}_{k-1 \text{ times}}, \underbrace{0, \ldots, 0}_{(n-r)k \text{ times}},\eta_r),$$
where
$x_1\geq\eta_1\geq x_2\geq\eta_2\geq\cdots\geq\eta_{r-1}\geq x_r\geq0\geq\eta_r$. By Theorem \ref{thm:ared}, we have $\rho_{v,\mu}\in\mathrm{ARED}_{n,k}$ if and only if
\begin{align*}
\sum_{i=1}^rx_i(k-1)\frac{\mu}{nk}+\sum_{i=1}^{r-1}\eta_i\frac{\mu}{nk}+\eta_r\left(\frac{\mu}{nk}+1-\mu\right)\geq0\quad\forall x\in \Delta_r.
\end{align*}
Now, by Corollary \ref{simplestatementofpurereduction}, we have $\eta_r+\sum_{i=1}^{r-1}\eta_i=0$.
Also,
$x\in \Delta_r$
implies $\sum_{i=1}^rx_i=1$.
Therefore, making the dependence of $\eta_r=\eta_r(x)$ on $x$ explicit, we can simplify:
\begin{align}
\rho_{v,\mu}\in\mathrm{ARED}_{n,k}\quad\Leftrightarrow\quad& (k-1)\frac{\mu}{nk}+(1-\mu)\eta_r(x)\geq0\quad\forall x \in \Delta_r
\nonumber\\
\Leftrightarrow\quad& \mu\geq\left(\frac{k-1}{-nk\eta_r(x)}+1\right)^{-1}\quad\forall x \in \Delta_r.
\label{ARED1largeEV}
\end{align}
The right-hand side of the last condition is monotonically decreasing in $\eta_r(x)\leq0$. Thus, the strongest constraint on $\mu$ is obtained for the smallest possible value of $\eta_r(x)$, which is $\eta_r=-\frac{r-1}{r}$ by Lemma \ref{MXlemma} (3), occurring exactly if $x_i=\frac{1}{r}~\forall i=1,\ldots,r$. Plugging this into Eq.\ (\ref{ARED1largeEV}) above gives the claim. Note that the optimal vectors $x$ correspond exactly to the maximally entangled states on $\mathbb C^n\otimes \mathbb C^k$.

(2) According to \cite[Lemma III.3]{hil}, the smallest eigenvalue of $$(\rho_{v,\mu})^\Gamma=\mu I_{nk}/(nk)+(1-\mu)(vv^*)^\Gamma$$ is $\frac{\mu}{nk}+(1-\mu)(-\sqrt{\nu_1\nu_2})$. This is non-negative iff the condition from the proposition holds. The condition for $\mathrm{SEP}_{n,k}$ follows from \cite[Appendix B]{vidaltarrach}.

(3) The orbit $\left\{U\rho_{v,\mu}U^*\,|\,U \in \mathcal U(nk)\right\}$ is exactly the set $$\left\{\mu I_{nk}/(nk)+(1-\mu)ww^*\,|\,w\in\mathbb C^n\otimes \mathbb C^k,\,\|w\|=1\right\}.$$ By the previous item, these states are all $\mathrm{PPT}$ if and only if $\mu\geq\left(\frac{1}{nk\sqrt{\omega_1\omega_2}}+1\right)^{-1}$ for all $1\geq\omega_1\geq\omega_2\geq0$ with $\omega_1+\omega_2\leq1$ (i.e.\ these are the largest two Schmidt coefficients of any one of the previous vectors $w$). The strongest constraint on $\mu$ is obtained for $\omega_1=\omega_2=1/2$, giving the desired result for $\mathrm{APPT}$. The statement for $\mathrm{ASEP}$ follows from the previous point.

(4) We have $(\rho_{v,\mu})^{red}=\mu I_{nk}(k-1)/(nk)+(1-\mu)(vv^*)^{red}$. The smallest eigenvalue of this matrix is, by Theorem \ref{thm:red-pure}, $\mu\frac{k-1}{nk}+(1-\mu)\eta_r(v)$, where $\eta_r(v)$ denotes the smallest solution of the equation $F_v(\eta)=0$ (cf.\ statement of Theorem \ref{thm:red-pure}). Therefore, $(\rho_{v,\mu})^{red}\geq0$ if and only if $\eta_r(v)\geq-\frac{\mu}{1-\mu}\frac{k-1}{nk}$. This is equivalent to $F_v(-\frac{\mu}{1-\mu}\frac{k-1}{nk})\leq0$, since $F_v(\eta)$ is strictly monotonically increasing in the interval $\eta\in(-\infty,0]$ from the negative value $-1$ at $\eta=-\infty$ to a non-negative value at $\eta=0$. Writing out the explicit form of $F_v$ yields the claim.
\end{proof}

\begin{remark}
Examining the condition of Proposition \ref{exampleprop}(4) for all choices of $v$, one sees (by the method of Lagrange multipliers enforcing normalization $\sum_{i=1}^r\nu_i=1$) that the strongest constraint is obtained for $\nu_i=1/r$ for all $i=1,\ldots,r$, i.e.\ when $v$ corresponds to a maximally entangled state. Plugging this in, one obtains the condition in Proposition \ref{exampleprop}(1), yielding another proof for it. This proof method is similar to our proof of (3) via (2), but the difference is that the ``most constraining state'' for the APPT condition (3) was a maximally entangled state on a \emph{2-dimensional} subspace (see also \cite{largestseparableball}).
\end{remark}

\begin{corollary}\label{inclusioncorollary}
Under the condition $n,k\geq2$, we have:
\begin{enumerate}
\item $\mathrm{ARED}_{n,k}\not\subseteq\mathrm{APPT}_{n,k}\quad\Leftrightarrow\quad k\geq3$.
\item $\mathrm{ARED}_{n,k}=\mathrm{APPT}_{n,k}=\mathrm{ASEP}_{n,k}\subseteq\mathrm{SEP}_{n,k}\subseteq\mathrm{PPT}_{n,k}\quad\Leftrightarrow\quad k=2$.
\end{enumerate}
\end{corollary}
\begin{proof}
Let first $k\geq3$ and $n\geq2$. Define $r:=\min(n,k)$, choose a unit vector $v\in\mathbb C^n\otimes\mathbb C^k$ with the Schmidt coefficients $\nu_1=\nu_2=1/2$ (cf.\ statement of Proposition \ref{exampleprop}), and let $\mu:=\left(\frac{k-1}{nk}\frac{r}{r-1}+1\right)^{-1}$. Then $\rho_{v,\mu}\in\mathrm{ARED}_{n,k}\setminus\mathrm{APPT}_{n,k}$ by Proposition \ref{exampleprop}, since
\begin{align*}
\mu=\left(\frac{k-1}{nk}\frac{r}{r-1}+1\right)^{-1}<\left(\frac{1}{nk\sqrt{\nu_1\nu_2}}+1\right)^{-1}=\left(\frac{2}{nk}+1\right)^{-1},
\end{align*}
as one easily verifies.

The equality $\mathrm{ARED}_{n,k}=\mathrm{APPT}_{n,k}=\mathrm{ASEP}_{n,k}$ for $k=2$ holds by Proposition \ref{firstqubitproposition}.
\end{proof}

\section{Intermission: $\mathrm{APPT}$, $\mathrm{GER}$, $\mathrm{ASEP}$ and $\mathrm{SEPBALL}$}
We continue our treatment of $\mathrm{ARED}$ in the next section, where we will introduce simple polyhedral upper and lower bounds on it. But here we pause to first discuss in more detail the sets $\mathrm{APPT}$ (see especially  \cite{hil}) and $\mathrm{ASEP}$ (see in particular \cite{largestseparableball}) coming from the partial transposition criterion and from separability itself, and sets $\mathrm{GER}$ and $\mathrm{SEPBALL}$ which will turn out to be lower approximations to them. Let us now define the latter two sets and make their meaning clear afterwards.
\begin{definition}\label{definesepball} \rm Given $n,k\geq2$, denote $r:=\min(n,k)$, and define the following two sets:
\begin{align*}\mathrm{SEPBALL}_{n,k}~&:=~\Big\{\rho\in D_{n,k}\,\Big|\,{\rm Tr}[\rho^2]\leq\frac{1}{nk-1}\Big\},\\
\mathrm{GER}_{n,k}~&:=~\Big\{\lambda\in\Delta_{nk}\,\Big|\,\sum_{i=1}^{r-1}\lambda^\downarrow_i~\leq~2\lambda^\downarrow_{nk}+\sum_{i=1}^{r-1}\lambda^\downarrow_{nk-i}\,\Big\}.\end{align*}
Note that, as described in Section \ref{preliminariessection}, we will freely identify $\mathrm{GER}_{n,k}$ as the subset of $D_{n,k}$ consisting of those quantum states with spectrum in $\mathrm{GER}_{n,k}$, and conversely $\mathrm{SEPBALL}_{n,k}$ as a subset of $\Delta_{nk}$.
\end{definition}
It has been proven in \cite{largestseparableball} that all states in $\mathrm{SEPBALL}_{n,k}$ are separable and that this set is in fact the largest Euclidean ball inside $D_{n,k}$ (Ref.\ \cite{zycz} already implies  that this is the largest ball of PPT states).
In fact, since its characterization depends only on spectral information, we even get the lower approximation $\mathrm{SEPBALL}_{n,k}\subseteq\mathrm{ASEP}_{n,k}$. The fact that there cannot be a larger ball of separable states inside $D_{n,k}$ can be understood by noting that $\mathrm{SEPBALL}_{n,k}$ contains states on the boundary of $D_{n,k}$, namely all states $\rho$ with spectrum $\lambda_\rho=(1,1,\ldots,1,0)/(nk-1)$. Below we will show that these states are actually the only rank-deficient states (i.e.\ are on the boundary of $D_{n,k}$) in $\mathrm{ASEP}_{n,k}$.

The designation $\mathrm{GER}$ in the foregoing definition alludes to the ``Gershgorin circle theorem'' \cite{hornjohnson}. The defining equation of $\mathrm{GER}$ is exactly the sufficient condition provided by Gershgorin's theorem for all of Hildebrand's $\mathrm{APPT}$ matrix inequalities \cite{hil} to be satisfied, as we show in the next theorem. We thus obtain an easily-checkable sufficient condition for membership in $\mathrm{APPT}$, which is in particular simpler than Hildebrand's condition \cite{hil} that involves checking the positivity of an exponential number (in $\min(n,k)$) of Hermitian matrices, but on the other hand is sufficient \emph{and} necessary.
\begin{theorem}\label{simplesufftheoremforAPPT}Let $\rho\in D_{n,k}$ (for $n,k\geq2$) with decreasingly ordered eigenvalue vector $\lambda$, and denote $r:=\min(n,k)$. Then: $\rho\in\mathrm{APPT}_{n,k}$ whenever
\begin{align}
\sum_{i=1}^{r-1}\lambda_i~\leq~2\lambda_{nk}+\sum_{i=1}^{r-1}\lambda_{nk-i}~.\label{necessaryforAPPTeqn}
\end{align}
In other words: $\mathrm{GER}_{n,k}\subseteq\mathrm{APPT}_{n,k}$.
\end{theorem}
\begin{proof}
Let $j_0,j_1,\ldots ,j_{2r-2}$ be $(2r-1)$ pairwise distinct elements of the set $\{1,2,\ldots,nk\}$. As $\lambda$ is assumed to be decreasingly ordered, we have
\begin{align}\label{diagonalelementminussumofabsolutevalues}
2\lambda_{j_0}-\sum_{i=1}^{r-1}|\lambda_{j_{i}}-\lambda_{j_{i+r-1}}|~\geq~2\lambda_{nk}+\sum_{i=1}^{r-1}(\lambda_{nk-i}-\lambda_{i})~\geq0
\end{align}
by Eq.\ (\ref{necessaryforAPPTeqn}). Now, for any matrix occurring in Hildebrand's $\mathrm{APPT}$ criterion \cite[Lemma III.3]{hil}, the difference between any diagonal element and the sum of absolute values of the other entries in the same row equals the left-hand side of (\ref{diagonalelementminussumofabsolutevalues}) for some choice of pairwise distinct $j_0,\ldots,j_{2r-2}$ (for illustration in the case $k=3\leq n$, see the matrices displayed in Eq.\ (\ref{hildebrandmatrices3n})). By the Gershgorin circle theorem \cite{hornjohnson}, the non-negativity of all such differences ensures the positive-semidefiniteness of all these (Hermitian) matrices. Hildebrand's result \cite[Lemma III.3]{hil} thus gives $\rho\in\mathrm{APPT}$.
\end{proof}

Note that the condition (\ref{necessaryforAPPTeqn}) cannot be sharpened by the above proof technique which relies on a combination of only Hildebrand's criterion with the Gershgorin circle theorem. This is because one of the rows in a matrix of Hildebrand's criterion will always be given by the assignment $j_i=nk-i$ for $0\leq i\leq r-1$ and $j_i=i-r+1$ for $r\leq i\leq 2r-2$ \cite[Lemma III.10]{hil}, and for this the Gershgorin condition is exactly Eq.\ (\ref{necessaryforAPPTeqn}).

Using this assignment in Hildebrand's criterion, we further obtain the following:
\begin{proposition}[see also {\cite[Proposition 1]{johnstonmay2014}}]\label{boundaryproposition}A state $\rho\in\mathrm{APPT}_{n,k}$ is rank-deficient  if and only if it has the following spectrum:
$${\rm spec}(\rho)=\Big(\frac{1}{nk-1},\frac{1}{nk-1},\ldots,\frac{1}{nk-1},0\Big).$$
As $\mathrm{SEPBALL}_{n,k}\subseteq\mathrm{ASEP}_{n,k}\subseteq\mathrm{APPT}_{n,k}$, this means that these states are also the only rank-deficient states in $\mathrm{ASEP}_{n,k}$.
\end{proposition}
\begin{proof}Let $\lambda$ be the decreasingly-ordered eigenvalue vector of a rank-deficient state $\rho\in\mathrm{APPT}_{n,k}$. Thus, we have $\lambda_{nk}=0$. As noted in the paragraph preceding the statement of the present proposition, one of the $r\times r$-matrices (where $r:=\min(n,k)$), for which Hildebrand's criterion \cite{hil} ensures positive-semidefiniteness, has the following first row:
$$(\lambda_{nk}=0,\lambda_{nk-1}-\lambda_1,\lambda_{nk-2}-\lambda_2,\ldots,\lambda_{nk-r+1}-\lambda_{r-1}).$$
As the diagonal element $\lambda_{nk}=0$ of this matrix vanishes, positive-semidefiniteness of the matrix enforces the entire corresponding row to vanish, so that in particular $\lambda_{nk-1}=\lambda_1$. Since $\lambda$ was assumed to be decreasingly ordered and normalized, we get that $$\lambda_1=\cdots=\lambda_{nk-1}=\frac{1}{nk-1}\,.$$ The fact that all such states are contained in $\mathrm{SEPBALL}_{n,k}\subseteq\mathrm{APPT}_{n,k}$, follows from the definition of $\mathrm{SEPBALL}_{n,k}$.
\end{proof}

\begin{remark}
Note that $\mathrm{GER}_{n,k}\subseteq\mathrm{APPT}_{n,k}$ is a polyhedral subset of $\Delta_{nk}$ containing the extreme points from Proposition \ref{boundaryproposition}, as is easily seen from its definition. $\mathrm{SEPBALL}_{n,k}\subseteq\mathrm{APPT}_{n,k}$ contains these boundary states as well, but the set is ``round'' due to its definition via Euclidean distances which is quadratic in the eigenvalues -- thus it has a unique supporting hyperplane at these boundary points, which coincides with a facet of $D_{n,k}$. Both these facts together imply that $\mathrm{SEPBALL}_{n,k}\not\subseteq\mathrm{GER}_{n,k}$, which can also be seen by explicit examples of states. Furthermore, it is $\mathrm{GER}_{n,k}\not\subseteq\mathrm{SEPBALL}_{n,k}$, which will for example follow from Proposition \ref{propsomegeometricquantities}. But since both $\mathrm{GER}_{n,k}$ and $\mathrm{SEPBALL}_{n,k}\subseteq\mathrm{SEP}_{n,k}$ are contained in $\mathrm{APPT}_{n,k}$, we have the following lower approximation to $\mathrm{APPT}_{n,k}$ which has the benefit of being much easier than the exact characterization given in \cite{hil}:
$$\mathrm{GER}_{n,k}\cup\mathrm{SEPBALL}_{n,k}\subseteq\mathrm{APPT}_{n,k}.$$
\end{remark}

\section{A family of intermediate criteria}\label{sec:A family of intermediate criteria}

For arbitrary $p \in [nk]$, let us introduce the sets of eigenvalue vectors for which the largest eigenvalue is less or equal than the sum of the $p$ smallest:
$$\mathrm{LS}_p:=\{\lambda \in \Delta_{nk} \, : \, \lambda_1^\downarrow \leq \lambda_{nk-p+1}^\downarrow + \lambda_{nk-p+2}^\downarrow + \cdots +\lambda_{nk}^\downarrow\}.$$
Obviously,  for $p <q$, $\mathrm{LS}_p \subseteq \mathrm{LS}_q$. Furthermore, one has $\mathrm{LS}_1 = \{\mathbf{1}_{nk}/(nk)\}$ and $\mathrm{LS}_{nk} = \Delta_{nk}$.

Let us now consider how these simple sets $\mathrm{LS}_q$ are positioned with respect to the sets $\mathrm{APPT}$ and $\mathrm{ARED}$:
\begin{theorem}\label{intermediatecriteriathm}
For $n,k\geq3$, we have
$$\mathrm{APPT} \subseteq \mathrm{LS}_3 \subseteq\mathrm{LS}_k \subseteq \mathrm{ARED}_{n,k} \subseteq \mathrm{LS}_{2k-1}.$$
More exactly, the following inclusions hold:
\begin{enumerate}
\item\label{APPTinASS3}For $n,k\geq2$: $\mathrm{APPT}\subseteq\mathrm{LS}_3$.
\item\label{APPT2in-notinAPPT2}For $\min(n,k)\in\{2,3\}$: $\mathrm{LS}_2\subseteq\mathrm{APPT}$.\\For $\min(n,k)\geq4$: $\mathrm{LS}_2\not\subseteq\mathrm{APPT}$.
\item\label{APPTnotinASS2}For $n,k\geq2$: $\mathrm{APPT}\not\subseteq\mathrm{LS}_2$.
\item\label{ASSkinARED}For $n,k\geq2$: $\mathrm{LS}_k\subseteq\mathrm{ARED}_{n,k}$.
\item\label{AREDinASS2kminus1}For $n,k\geq2$: $\mathrm{ARED}_{n,k}\subseteq\mathrm{LS}_{2k-1}$.
\end{enumerate}
\end{theorem}
\begin{proof}Ad (\ref{APPTinASS3}): Let us start with the first inclusion and consider a pure vector $x \in \mathbb C^n \otimes \mathbb C^k$, having two non-zero Schmidt coefficients, both equal to $1/2$. By Hildebrand's criterion \cite{hil}, it follows that any $\lambda \in \mathrm{APPT}$ must satisfy
$$\langle \lambda^\downarrow, (-1/2, \underbrace{0, \ldots, 0}_{nk-4 \text{ times}}, 1/2,1/2,1/2) \rangle \geq 0,$$
which is exactly the condition $\lambda \in \mathrm{LS}_{3}$.

Ad (\ref{APPT2in-notinAPPT2}): Here we may w.l.o.g.\ assume $\min(n,k)=k$. For $k=2$, the assertion follows from the following inequality together with Proposition \ref{firstqubitproposition}:
$$\lambda_1 \leq \lambda_{2n-1} + \lambda_{2n} \leq \lambda_{2n-1} + 2\sqrt{\lambda_{2n-2}\lambda_{2n}}.$$

For $k=3\leq n$ and $\lambda\in\mathrm{LS}_2$ (with decreasingly ordered components), we show that $\lambda\in\mathrm{APPT}$, which by the criterion given in \cite[Corollary V.3]{hil} is equivalent to the following two matrices being positive-semidefinite:
\begin{equation}\label{hildebrandmatrices3n}
\begin{pmatrix}2\lambda_{3n}&\lambda_{3n-1}-\lambda_1&\lambda_{3n-2}-\lambda_2\\
\lambda_{3n-1}-\lambda_1&2\lambda_{3n-3}&\lambda_{3n-4}-\lambda_3\\
\lambda_{3n-2}-\lambda_2&\lambda_{3n-4}-\lambda_3&2\lambda_{3n-5}\end{pmatrix},
~~~\begin{pmatrix}2\lambda_{3n}&\lambda_{3n-1}-\lambda_1&\lambda_{3n-3}-\lambda_2\\
\lambda_{3n-1}-\lambda_1&2\lambda_{3n-2}&\lambda_{3n-4}-\lambda_3\\
\lambda_{3n-3}-\lambda_2&\lambda_{3n-4}-\lambda_3&2\lambda_{3n-5}\end{pmatrix}.
\end{equation}

A sufficient condition for a Hermitian matrix to be positive-semidefinite is, by the Gershgorin circle theorem \cite{hornjohnson}, for each diagonal element to be at least as large as the sum of the absolute values of the other entries in the same row. By the assumed ordering of the entries of $\lambda$, this condition is the most constraining for the first row of the first of both matrices. Considering this row, we have:
\begin{align*}\label{hildebrandmatrices3n}
2\lambda_{3n}-|\lambda_{3n-1}-\lambda_1|-|\lambda_{3n-2}-\lambda_2|&=(\lambda_{3n}+\lambda_{3n-1}-\lambda_1)+(\lambda_{3n}+\lambda_{3n-2}-\lambda_2)\\
&\geq2(\lambda_{3n}+\lambda_{3n-1}-\lambda_1)\geq0,
\end{align*}
where the last inequality follows from $\lambda\in\mathrm{LS}_2$, and finally implies $\lambda\in\mathrm{APPT}$.

For the case $k\geq4$ consider the following element, with $a:=1/[nk+k(k-1)/2]>0$:
$$\lambda = (\underbrace{2a, \ldots, 2a}_{k(k-1)/2 \text{ times}}, a, \ldots, a ) \in \Delta_{nk}.$$
Obviously, $\lambda \in \mathrm{LS}_2$. However, for a maximally entangled vector $x=\sum_{i=1}^ke_i\otimes  f_i/\sqrt{k} \in \mathbb C^n \otimes \mathbb C^k$, the partially-transposed operator $(xx^*)^{\Gamma}$ has eigenvalues $(-1/k)$ with multiplicity $k(k-1)/2$ and $(1/k)$ with multiplicity $k(k+1)/2$ and $0$ otherwise. Thus, $\lambda\in\mathrm{APPT}$ would imply \cite{hil}:
$$\lambda \in \mathrm{APPT} \implies 0\leq-\frac{1}{k}\frac{k(k-1)}2 2a +\frac{1}{k} \frac{k(k+1)}2 a,$$
which is false for $k >3$, finishing the proof.

Ad (\ref{APPTnotinASS2}): This will follow from Proposition \ref{propsomegeometricquantities} (from the statement $\Lambda(\mathrm{APPT})>\Lambda(\mathrm{LS}_2)$).

Ad (\ref{ASSkinARED}): For arbitrary vectors $\lambda = \lambda^\downarrow \in \Delta_{nk}$ and $x  =x^\downarrow \in \Delta_r$, where $r:=\min(n,k)$, we use the bound $\eta_i \leq x_i$ (for $i \in [r-1]$), the equality $-\eta_{r}=\eta_1+\ldots+\eta_{r-1}$ (see Lemma \ref{MXlemma}(\ref{MXlemmatrace0}) and Definition \ref{def:hat-xi}), and the ordering of the vectors $\lambda$ and $x$ to get
\begin{align*}
\langle \lambda^\downarrow,\hat x^\uparrow \rangle &= x_1 (\lambda_{nk} + \cdots + \lambda_{nk-k+2}) + \eta_1 \lambda_{nk-k+1} \\
&\quad+ x_2 (\lambda_{(n-1)k} + \cdots + \lambda_{(n-1)k-k+2}) + \eta_2 \lambda_{(n-1)k-k+1} + \cdots \\
&\quad+ x_r (\lambda_{(n-r+1)k} + \cdots + \lambda_{(n-r+1)k-k+2}) + 0 \cdot (\lambda_{(n-r)k+1} + \cdots + \lambda_2) + \eta_r \lambda_1\\
&\geq x_1 (\lambda_{nk} + \cdots + \lambda_{nk-k+2}) + x_1 \lambda_{nk-k+1} \\
&\quad+ x_2 (\lambda_{(n-1)k} + \cdots + \lambda_{(n-1)k-k+2}) + x_2 \lambda_{(n-1)k-k+1} + \cdots \\
&\quad+ x_r (\lambda_{(n-r+1)k} + \cdots + \lambda_{(n-r+1)k-k+2}) + 0 \cdot (\lambda_{(n-r)k+1} + \cdots + \lambda_2) + \eta_r' \lambda_1,
\end{align*}
where $\eta_r'$ is chosen as follows:
$$\eta_r' = -(x_1 + x_2 + \cdots + x_{r-1}) = x_r -1.$$
We continue towards a concise lower bound for $\langle \lambda^\downarrow,\hat x^\uparrow \rangle$, using $\sum_{i=1}^rx_i=1$:
\begin{align*}
\langle \lambda^\downarrow,\hat x^\uparrow \rangle &\geq x_1 (\lambda_{nk} + \cdots + \lambda_{nk-k+1})  \\
&\quad+ x_2 (\lambda_{(n-1)k} + \cdots + \lambda_{(n-1)k-k+1}) + \cdots\\
&\quad+ x_r (\lambda_{(n-r+1)k} + \cdots + \lambda_{(n-r+1)k-k+2} + \lambda_1)  - \lambda_1\\
&\geq (\lambda_{nk} + \cdots + \lambda_{nk-k+1}) - \lambda_1.
\end{align*}
For any fixed $\lambda \in \mathrm{LS}_k$, this last expression is non-negative, implying that $\langle \lambda^\downarrow,\hat x^\uparrow \rangle\geq0$ for any $x  =x^\downarrow \in \Delta_r$, so that $\lambda\in\mathrm{ARED}_{n,k}$ by Theorem \ref{thm:ared}.

Ad (\ref{AREDinASS2kminus1}): This is simply the constraint for $r=2$ from Corollary \ref{cor.ared}.
\end{proof}

\bigskip
Now we look at some geometrical quantities associated to the various sets used and defined earlier in this section. In particular, for any subset $A \subseteq \Delta_{nk}$, we define
$$\Lambda(A) ~:=~ \sup_{\lambda \in A} \lambda_1^\downarrow.$$
When $A$ is identified with the set of spectra of a set $A\subseteq D_{n,k}$, then $\Lambda(A)=\sup_{\rho\in A}\|\rho\|_\infty$.
Furthermore, since the function $\lambda\mapsto\lambda_1^\downarrow$
is convex (similarly, $\|\rho\|_\infty$ and ${\rm Tr}[\rho^2]$), when $A$ is convex the supremum in the definitions of $\Lambda(A)$ is attained at an extreme point of $A$.
Note also that, when $\Lambda(A)\geq 2/(nk)$,  $\Lambda(A)-1/(nk)$ is the radius of the smallest operator-norm ball around $I_{nk}/(nk)$ containing $A$.

\begin{proposition}[Some geometrical quantities]\label{propsomegeometricquantities}
Let $n,k\geq2$. Then we have:
\begin{align*}
\Lambda(\mathrm{ARED})~&=~
\begin{cases}
\frac{k+1}{k(n+1)}&\quad\text{ if }~~k\leq n\\
\frac{1}{n}&\quad\text{ if }~~k\geq n
\end{cases}\\
\Lambda(\mathrm{LS}_p)~&=~ \frac{p}{nk+p-1}\\
\Lambda(\mathrm{ASEP})~=~\Lambda(\mathrm{APPT})~&=~\frac{3}{2+nk}\\
\Lambda(\mathrm{GER})~&=~\frac{3}{2+nk}\\
\Lambda(\mathrm{SEPBALL})~&=~\frac{2}{nk}~.
\end{align*}
\end{proposition}
\begin{proof}
That $\Lambda(\mathrm{ARED})$ is at least the given expression follows from  Proposition \ref{exampleprop}(1) by setting $\mu=\left(\frac{k-1}{nk}\frac{r}{r-1}+1\right)^{-1}$ (here, $r=\min(n,k)$) and calculating the largest eigenvalue of $\rho_{v,\mu}$. The converse inequality follows from the constraint $\langle \lambda^\downarrow,\hat x^\uparrow \rangle \geq 0$ in Theorem \ref{thm:ared} for a maximally entangled vector $x$. Specifically, for $x$ with Schmidt coefficients $x_1=\cdots=x_r=1/r$ and for a normalized and decreasingly ordered eigenvalue vector $\lambda$, we have due to $\eta_i=\frac{1}{r}$ for $i=1,\ldots,r-1$ and $\eta_{r}=-\frac{r-1}{r}$ (cf.\ proof of Proposition \ref{exampleprop}(1)):
\begin{align*}
\langle \lambda,\hat x^\uparrow \rangle=-\frac{r-1}{r}\lambda_1+\frac{1}{r}\sum_{i=(n-r)k+2}^{nk}\lambda_i.
\end{align*}
The sum in the last expression is not greater than $$\sum_{i=(n-r)k+2}^{nk}\lambda_{(n-r)k+2}=(rk-1)\lambda_{(n-r)k+2},$$ and this sum is also never greater than $$1-\lambda_1-\sum_{i=2}^{(n-r)k+1}\lambda_i\leq1-\lambda_1-(n-r)k\lambda_{(n-r)k+2}.$$ The sum is therefore never greater than $$\max_{x\in[0,1]}\min\{(rk-1)x,1-\lambda_1-(n-r)kx\}=(1-\lambda_1)\frac{rk-1}{nk-1},$$ and thus:
\begin{align*}
\langle \lambda,\hat x^\uparrow \rangle\leq-\frac{r-1}{r}\lambda_1+\frac{1}{r}(1-\lambda_1)\frac{rk-1}{nk-1}=\frac{rk-1}{r(nk-1)}\left[1-\lambda_1\frac{nkr+kr-r-nk}{kr-1}\right].
\end{align*}
This last expression has to be non-negative if $\lambda\in \mathrm{ARED}_{n,k}$, and thus $\lambda_1\leq\frac{kr-1}{nkr+kr-r-nk}$, which is the announced result, depending on the value of $r$.

Let us now show the bound for $\mathrm{LS}_p$. First, we have the upper bound
$$\lambda_1 \leq \lambda_{nk-p+1} + \cdots +\lambda_{nk} \leq p \lambda_{nk-p+1}.$$
We also have
$$\lambda_1 \leq \lambda_{nk-p+1} + \cdots +\lambda_{nk} = 1 - \lambda_1 - \sum_{i=2}^{nk-p} \lambda_i \leq 1 - \lambda_1 - (nk-p-1)\lambda_{nk-p+1}.$$
Putting the two inequalities together, we get
$$\lambda_1 \leq \min\left\{p \lambda_{nk-p+1}, 1 - \lambda_1 - (nk-p-1)\lambda_{nk-p+1}\right\} \leq p\frac{1-\lambda_1}{nk-1},$$
which shows that $\Lambda(\mathrm{LS}_p)$ is at most $p/(nk+p-1)$. To show that the bound is attained, one needs to consider, for a suitable normalizing $a>0$, a vector of the form $(pa, a, \ldots, a)$.

To prove the statement for $\mathrm{APPT}$ and $\mathrm{GER}$, note that $\mathrm{GER}\subseteq\mathrm{APPT}\subseteq\mathrm{LS}_3$ by Theorems \ref{intermediatecriteriathm} and \ref{simplesufftheoremforAPPT}, which implies $\Lambda(\mathrm{GER})\leq\Lambda(\mathrm{APPT})\leq\Lambda(\mathrm{LS}_3)=3/(2+nk)$, as shown above. On the other hand, the state $\rho_{v,\mu}$ from the statement of Proposition \ref{exampleprop}(3) with $\mu=(\frac{2}{nk}+1)^{-1}$ is easily seen to be an element of $\mathrm{GER}$, since for $i=2,\ldots,nk$ it is $\lambda_1^\downarrow(\rho_{v,\mu})=3\lambda_i^\downarrow(\rho_{v,\mu})$. Thus, $\rho_{v,\mu}\in\mathrm{GER}$, so that one gets $\Lambda(\mathrm{GER})\geq\lambda_1^\downarrow(\rho_{v,\mu})=3/(2+nk)$.

From the previous paragraph, the maximum in the definition of $\Lambda(\mathrm{APPT})$ is attained at a pseudo-pure state. But any such state is APPT if and only if it is ASEP \cite[Appendix B]{vidaltarrach} (see also Remark \ref{ASEPoutsideSEPBALL}). This shows $\Lambda(\mathrm{ASEP})\geq\Lambda(\mathrm{APPT})$, which together with the trivial statement $\Lambda(\mathrm{ASEP})\leq\Lambda(\mathrm{APPT})$ gives the value of $\Lambda(\mathrm{ASEP})$.

To show the statement for $\mathrm{SEPBALL}$, consider $\lambda_1\in[0,1]$ to be fixed. Then, by strict convexity, the expression ${\rm Tr}[\rho^2]=\sum_{i=1}^{nk}\lambda_i^2$ is uniquely minimized under the normalization constraint $\sum_{i=2}^{nk}\lambda_i=1-\lambda_1$ by the assignment $\lambda_i:=(1-\lambda_1)/(nk-1)$ for $i=2,\ldots,nk$. Thus, $\Lambda(\mathrm{SEPBALL})$ equals the largest solution $\lambda_1$ of $\lambda_1^2+(nk-1)[(1-\lambda_1)/(nk-1)]^2=1/(nk-1),$ which is $\lambda_1=2/(nk)$.
\end{proof}

\begin{remark}\label{ASEPoutsideSEPBALL}Proposition \ref{propsomegeometricquantities} implies the existence of absolutely separable state outside of the largest separable ball when $nk\geq6$ and $n,k\geq2$, since in this case $\Lambda(\mathrm{ASEP}_{n,k})>\Lambda(\mathrm{SEPBALL})$.
\end{remark}

\begin{remark}Proposition \ref{propsomegeometricquantities} shows that $\mathrm{SEPBALL}\not\subseteq\mathrm{LS}_2$. Considering a state $\rho\in D_{n,k}$ with $3$ eigenvalues $2/(nk+3)$ and $(nk-3)$ eigenvalues $1/(nk+3)$ shows that $\mathrm{LS}_2\not\subseteq\mathrm{SEPBALL}$ if $nk\geq10$.\end{remark}

\section{Decompositions of different dimensions}\label{sectiondecompofdifferentdimensions}

As in Section $IV$ of \cite{hil}, we would like, for two different tensor decompositions
$$\mathbb C^d = \mathbb C^{n_1} \otimes \mathbb C^{k_1} = \mathbb C^{n_2} \otimes \mathbb C^{k_2},$$
to compare the sets $\mathrm{ARED}_{n_1,k_1}$ and $\mathrm{ARED}_{n_2,k_2}$.

\begin{theorem}\label{thm:different-dimensions}
Consider two different tensor decompositions of $\mathbb C^d$, given by $d=n_1k_1=n_2k_2$, such that $\min(n_1,k_1) \geq \min(n_2,k_2)$ and $k_1 \leq k_2$. Then,
$$\mathrm{ARED}_{n_1,k_1} \subseteq \mathrm{ARED}_{n_2,k_2}.$$
\end{theorem}
\begin{proof}
Denote  $r_i=\min(n_i,k_i)$, $i=1,2$. Let $\rho\in \mathrm{ARED}_{n_1,k_1}$ and let $x=(x_1,\ldots,x_{r_2}) \in \Delta_{r_2}$ be a vector of Schmidt rank $r$. Since $\Delta_{r_1} \supseteq \Delta_{r_2}\oplus 0^{r_2-r_1}$, all we have to check is that
$$\langle \lambda_{\rho}^{\downarrow},\hat y^\uparrow \rangle \leq \langle \lambda_{\rho}^{\downarrow},\hat x^\uparrow \rangle,$$
where $y=(x_1,\ldots,x_{r_2}, \underbrace{0,\ldots,0}_{r_1-r_2 \text{ times}}) \in \Delta_{r_1}$ and $\lambda_{\rho}^{\downarrow}$ is the vector of eigenvalues of $\rho$ ordered decreasingly.
By Definition \ref{def:hat-xi},  we have
$$\hat x ~=~ (\underbrace{x_1, \ldots, x_1}_{k_2-1 \text{ times}}, \eta_1 ,  \underbrace{x_2, \ldots, x_2}_{k_2-1 \text{ times}}, \ldots, \eta_{r-1}, \underbrace{x_r, \ldots, x_r}_{k_2-1 \text{ times}}, \underbrace{0, \ldots, 0}_{(n_2-r)k_2 \text{ times}},\eta_r)$$
and
$$\hat y ~=~ (\underbrace{x_1, \ldots, x_1}_{k_1-1 \text{ times}}, \eta_1 ,  \underbrace{x_2, \ldots, x_2}_{k_1-1 \text{ times}}, \ldots, \eta_{r-1}, \underbrace{x_r, \ldots, x_r}_{k_1-1 \text{ times}}, \underbrace{0, \ldots, 0}_{(n_1-r)k_1 \text{ times}},\eta_r).$$
We have then,
\begin{align*}
\langle \lambda_{\rho}^{\downarrow},\hat x^\uparrow \rangle &= x_1 (\lambda_{n_2k_2} + \cdots + \lambda_{n_2k_2-k_2+2}) \\
&\quad+ \eta_1 \lambda_{n_2k_2-k_2+1} \\
&\quad+ x_2 (\lambda_{(n_2-1)k_2} + \cdots + \lambda_{(n_2-1)k_2-k_2+2}) \\
&\quad+ \eta_2 \lambda_{(n_2-1)k_2-k_2+1} \\
&\quad +\ldots \\
&\quad+ x_r (\lambda_{(n_2-r+1)k_2} + \cdots + \lambda_{(n_2-r+1)k_2-k_2+2}) \\
&\quad+ 0 (\lambda_{(n_2-r)k_2+1} + \cdots + \lambda_2) \\
&\quad+ \eta_r \lambda_1
\end{align*}
and
\begin{align*}
\langle \lambda_{\rho}^{\downarrow},\hat y^\uparrow \rangle &= x_1 (\lambda_{n_1k_1} + \cdots + \lambda_{n_1k_1-k_1+2}) \\
&\quad+ \eta_1 \lambda_{n_1k_1-k_1+1} \\
&\quad+ x_2 (\lambda_{(n_1-1)k_1} + \cdots + \lambda_{(n_1-1)k_1-k_1+2}) \\
&\quad+ \eta_2 \lambda_{(n_1-1)k_1-k_1+1} \\
&\quad +\ldots \\
&\quad+ x_r (\lambda_{(n_1-r+1)k_1} + \cdots + \lambda_{(n_1-r+1)k_1-k_1+2}) \\
&\quad+ 0 (\lambda_{(n_1-r)k_1+1} + \cdots + \lambda_2) \\
&\quad+ \eta_r \lambda_1.
\end{align*}
Since the sum multiplying $x_i$ ($i=1,\ldots,r$) in the expression  $\langle \lambda_{\rho}^{\downarrow},\hat y^\uparrow \rangle$ contains at most the same number of non-negative terms as the one from $\langle \lambda_{\rho}^{\downarrow},\hat x^\uparrow \rangle$ and each of these terms corresponds to a greater one in the sum multiplying $x_i$ from the expression $\langle \lambda_{\rho}^{\downarrow},\hat x^\uparrow \rangle$, we obtain that  $\langle \lambda_{\rho}^{\downarrow},\hat y^\uparrow \rangle \leq \langle \lambda_{\rho}^{\downarrow},\hat x^\uparrow \rangle.$ Thus, since $\langle \lambda_{\rho}^{\downarrow},\hat y^\uparrow \rangle\geq 0$, we obtain the conclusion.
\end{proof}

\begin{corollary}\label{cor:nk-kn}
For any $n \geq k$, we have that $\mathrm{ARED}_{n,k} \subseteq \mathrm{ARED}_{k,n}$.
\end{corollary}
Intuitively, this result can be understood in the following way. We have $\mathrm{APPT}_{n,k}=\mathrm{APPT}_{k,n}$ (see Section \ref{preliminariessection}), whereas in computing $\mathrm{ARED}$ a completely positive map $X\mapsto I\cdot{\rm Tr}[X]-X^T$ is applied after application of the transposition $\Theta$, and this completely positive map ``messes up'' the entanglement test the more the larger the dimension is to which it is applied. That this intuition holds only on the level of $\mathrm{ARED}$ but not on the level of $\mathrm{RED}$ can be seen by Proposition \ref{propositionREDvsPPT} (\ref{inpropredredprimegeq3},\ref{inpropredredprimefrom2}).

Let us now turn to the case when $k_1>k_2$.

\begin{proposition}\label{p.k1>k2}
Consider two different tensor decomposition of $\mathbb{C}^{d}$ given by $d=n_1k_1=n_2k_2$. If $n_1\geq k_1\geq 2$, $n_2\geq k_2\geq 2$ and $k_1>k_2$, then
$$ \mathrm{ARED}_{n_1,k_1}\nsubseteq \mathrm{ARED}_{n_2,k_2}.$$
On the other hand, if in addition $k_1\geq 2 k_2-1$, then
$$ \mathrm{ARED}_{n_2,k_2}\subseteq \mathrm{ARED}_{n_1,k_1}.$$
\end{proposition}
\begin{proof}
Since $k_1>k_2$ and $n_1k_1=n_2 k_2$, it follows that $n_1<n_2$, which is equivalent to $\frac{n_1}{n_1+1}<\frac{n_2}{n_2+1}$. Hence, we can choose $\mu\in\left[\frac{n_1}{n_1+1}, \frac{n_2}{n_2+1} \right)$. Let
$$\rho_{v,\mu}:=\mu I_{d}/d+(1-\mu)vv^*,$$
where $v\in\mathbb{C}^{d}$ with $\|v\|=1$. A simple computation shows that
$$ \left(\frac{k_1-1}{n_1k_1}\frac{k_1}{k_1-1}+1\right)^{-1}\leq\mu<\left(\frac{k_2-1}{n_2k_2}\frac{k_2}{k_2-1}+1\right)^{-1}$$
and thus, by Proposition \ref{exampleprop}(1) it follows that $\rho_{v,\mu}\in \mathrm{ARED}_{n_1,k_1}\setminus \mathrm{ARED}_{n_2,k_2}.$ On the other hand, by Theorem \ref{intermediatecriteriathm}, we have
$$ \mathrm{ARED}_{n_2,k_2}\subseteq \mathrm{LS}_{2k_2-1}\subseteq \mathrm{LS}_{k_1}\subseteq \mathrm{ARED}_{n_1,k_1}, \text{ if } k_1\geq 2k_2-1.$$
\end{proof}
We leave the analogous inclusions not covered by Theorem \ref{thm:different-dimensions} or Proposition \ref{p.k1>k2} as open questions, see Problem \ref{problemHildebrandDecompositions} below.

If we consider $\mathrm{RED}$ instead of $\mathrm{ARED}$, then there is generally \emph{no} inclusion as in Corollary \ref{cor:nk-kn}, as we show now. For this statement, recall from Section \ref{preliminariessection} that $\mathrm{RED}'_{n,k}$ denotes the set defined similarly as $\mathrm{RED}_{n,k}$, but with the reduction map acting on the \emph{first} tensor factor (i.e.\ of Hilbert space dimension $n$). Note that, with analogous notation, Corollary \ref{cor:nk-kn} can be expressed by saying $\mathrm{ARED}_{n,k} \subseteq \mathrm{ARED}'_{n,k}$ for $k\leq n$.
\begin{proposition}\label{propositionREDvsPPT}The following relations hold:
\begin{enumerate}
\item\label{inpropredredprimegeq3}For $n\geq3$, $k\geq3$: $\mathrm{RED}'_{n,k}\not\subseteq\mathrm{RED}_{n,k}\not\subseteq\mathrm{RED}'_{n,k}$.
\item\label{inpropredredprimefrom2}For $n\geq3$, $k=2$: $\mathrm{RED}'_{n,k}\not\subseteq\mathrm{RED}_{n,k}=\mathrm{PPT}_{n,k}$.
\item For $n=2$, $k=2$: $\mathrm{RED}_{n,k}=\mathrm{RED}'_{k,n}=\mathrm{PPT}_{n,k}$.
\end{enumerate}
\end{proposition}
\begin{proof}The equalities in items (2) and (3) follow from the fact that the reduction map on a subsystem of dimension 2 detects the same states as the transposition map (see Section \ref{preliminariessection}). To show the non-inclusion in item (2) for the case $n=3$, consider the following state $\rho_{3,2}\in D_{3,2}$:
\begin{align}\nonumber
\rho_{3,2}~:=~\frac{1}{1000}\left(
\begin{smallmatrix}
 110 & 30-39 i & 40-81 i & 48+37 i & 70-15 i & 12+i \\
 30+39 i & 128 & 66-i & 42-33 i & 134+5 i & 18-11 i \\
 40+81 i & 66+i & 174 & 28+73 i & 96+29 i & 30+47 i \\
 48-37 i & 42+33 i & 28-73 i & 188 & 110-13 i & 40+i \\
 70+15 i & 134-5 i & 96-29 i & 110+13 i & 226 & 48+47 i \\
 12-i & 18+11 i & 30-47 i & 40-i & 48-47 i & 174 \\
\end{smallmatrix}
\right)\,,
\end{align}
written here w.r.t.\ the product basis $\{|1,1\rangle,|1,2\rangle,|2,1\rangle,|2,2\rangle,|3,1\rangle,|3,2\rangle\}$. Then one finds numerically that $\rho_{3,2}^{red'}\geq I/20$ whereas $\rho_{3,2}^{red}\not\geq-I/20$, which implies $\rho_{3,2}\in\mathrm{RED}'_{3,2}\setminus\mathrm{RED}_{3,2}$. Now, for any given $n\geq3$ and $k\geq2$, one can simply embed the Hilbert space belonging to first subsystem of $\rho_{3,2}$ as a $3$-dimensional subspace into $\mathbb{C}^n$ and the Hilbert space belonging to the second subsystem of $\rho_{3,2}$ as a $2$-dimensional subspace into $\mathbb{C}^k$, and define the state $\rho_{n,k}\in D_{n,k}$ to agree with the action of $\rho_{3,2}$ on the tensor product of these two subspaces. By this embedding, the fact that $\rho_{3,2}\in\mathrm{RED}'_{3,2}\setminus\mathrm{RED}_{3,2}$ immediately implies $\rho_{n,k}\in\mathrm{RED}'_{n,k}\setminus\mathrm{RED}_{n,k}$, which proves the left non-inclusions in items (1) and (2). The right non-inclusion in item (1) follows by a swap of both subsystems.\end{proof}The non-inclusions from Proposition \ref{propositionREDvsPPT} (1,2) are already hinted at in the original works \cite{horodeckireduction,cerf}, albeit without explicit examples.

\section{Remarks and open questions}\label{remarkssection}

We would like to conclude our work with a series of comments and questions we leave open.

When comparing the results in the current paper for the set $\mathrm{ARED}$ with the ones for $\mathrm{APPT}$ developed in \cite{hil}, one notices immediately that Hildebrand characterizes $\mathrm{APPT}$ by a finite list of matrix inequalities, whereas our Theorem \ref{thm:ared} provides necessary and sufficient conditions as an infinite list of scalar, linear inequalities. From a practical point of view, it would be desirable to have a finite characterization of $\mathrm{ARED}$, so we leave open the following important question.

\begin{problem}
 Provide a \emph{finite} list of necessary and sufficient conditions for $\lambda \in \mathrm{ARED}_{n,k}$.
 \end{problem}

Continuing the parallel with the results in \cite{hil}, our Proposition \ref{p.k1>k2} leaves some cases open. Indeed, in \cite[Theorem IV.2]{hil}, the author shows the following inclusion:
$$\mathrm{APPT}_{n_1,k_1} \subseteq \mathrm{APPT}_{n_2,k_2},$$
whenever $\min(n_1,k_1) =: r_1 \geq r_2:=\min(n_2,k_2)$. In Proposition \ref{p.k1>k2} we show that the reversed inclusion holds for the sets $\mathrm{ARED}$,
$$ \mathrm{ARED}''_{n_1,k_1} \supseteq \mathrm{ARED}''_{n_2,k_2},$$
under the more restrictive condition $r_1 \geq 2r_2-1$. We believe that this condition is unnecessary.

 \begin{problem}\label{problemHildebrandDecompositions}
Show that, whenever $r_1 \geq r_2$, the previous inclusion holds.
 \end{problem}

\end{document}